\documentclass{article}

\pdfoutput=1

\usepackage[T1]{fontenc}
\usepackage{amsfonts}
\usepackage{amsmath}
\usepackage{mathrsfs}
\usepackage{amssymb}
\usepackage{amscd}
\usepackage{amsthm}
\usepackage{times}
\usepackage{multirow}
\usepackage{dsfont}
\usepackage{url}
\usepackage{bm}
\usepackage{pgf}
\usepackage{tikz}
\usepackage{slashed}
\usepackage[pdftex]{hyperref}
\usepackage{array}
\usepackage{natbib}
\usepackage{cancel}
\usepackage{stackrel}
\usepackage{graphicx}
\usepackage[normalem]{ulem}
\usetikzlibrary{arrows,automata,snakes}
\DeclareGraphicsExtensions{.ps,.eps,.pdf}
\usepackage{extarrows}

\newtheorem{subproperty}{Property}
\newtheorem{proposition}{Proposition}
\newtheorem{lemma}{Lemma}
\newtheorem{theorem}{Theorem}
\newtheorem{corollary}{Corollary}

\theoremstyle{definition}
\newtheorem{remark}{Remark}

\newtheorem{definition}{Definition}

\makeatletter
\def\Eqlfill@{\arrowfill@\Relbar\Relbar\Relbar}
\newcommand{\longmodels}[1][]{\,|\!\!\!\ext@arrow 0359\Eqlfill@{#1}}
\makeatother

\usepackage{authblk}

\newcommand{\pair}[2]{(#1, #2)}

\newcommand{\triple}[3]{( #1,#2,#3 )}

\newcommand{\Nat}{\mathbb{N}}
\newcommand{\Zed}{\mathbb{Z}}

\newcommand{\D}{\mathcal{D}}

\newcommand{\U}{\mathbf{U}}
\newcommand{\Snc}{\mathbf{S}}
\newcommand{\T}{\mathbf{T}}
\newcommand{\R}{\mathbf{R}}
\newcommand{\N}{\mathbb{N}}
\newcommand{\Z}{\mathbf{Z}}
\newcommand{\Real}{\mathbb{R}}
\newcommand{\Q}{\mathbb{Q}}

\newcommand{\X}{\mathbf{X}}
\newcommand{\Y}{\mathbf{Y}}

\newcommand{\Lng}{\mathscr {L}}
\newcommand{\iFF}{\text{ iff }}

\newcommand{\zot}{$\mathds{Z}$ot}
\newcommand{\G}{\mathbf{G}}
\newcommand{\F}{\mathbf{F}}
\newcommand{\siff}{\iFF}
\newcommand{\A}{\mathcal{A}}
\newcommand{\aX}{\mathrm{X}}
\newcommand{\aY}{\mathrm{Y}}

\newcommand{\step}[1]{\xrightarrow{#1}}

\newcommand{\pspace}{\textsc{PSPACE}}

\newcommand{\interval}[2]{[#1,#2]}

\newcommand{\PA}{{\rm PA}}
\newcommand{\QFP}{{\rm QFP}}
\newcommand{\DL}{{\rm DL}}
\newcommand{\IPC}{{\rm IPC}}

\newcommand{\cltl}[2]{{\rm CLTL_{#1}^{#2}}}
\newcommand{\LTL}{{\rm LTL}}

\newcommand{\symodels}{\longmodels{\mbox{\it{\tiny sym}}}}

\newcommand{\set}[1]{\{ #1 \}}

\newcommand{\weak}[1]{\overline{#1}}

\newcommand{\fresh}{\mathit{fresh}}
\newcommand{\terms}{\mathit{terms}}
\newcommand{\shift}{\mathit{shift}}
\newcommand{\shiftleft}{\mathit{sl}}  
\newcommand{\noprop}{\mathit{np}}

\newcommand{\sv}{\mathit{sv}}
\newcommand{\var}{\mathit{var}}
\newcommand{\const}{\mathit{const}}

\newcommand{\locfwd}{\preccurlyeq}
\newcommand{\locfwds}{\prec}
\newcommand{\locbwd}{\succcurlyeq}
\newcommand{\locbwds}{\succ}
\newcommand{\fwd}{\stackrel{\preccurlyeq}{\frown}}
\newcommand{\fwds}{\stackrel{\prec}{\frown}}
\newcommand{\bwd}{\stackrel{\succcurlyeq}{\frown}}
\newcommand{\bwds}{\stackrel{\succ}{\frown}}
\newcommand{\pfwd}[1]{\stackrel{#1}{\frown}}

\newcommand{\Predfwd}{\bm{F}^{\leq}}
\newcommand{\Predfwds}{\bm{F}^{<}}
\newcommand{\Predlocfwd}{\bm{f}^{\leq}}
\newcommand{\Predlocfwds}{\bm{f}^{<}}
\newcommand{\Predbwd}{\bm{B}^{\geq}}
\newcommand{\Predbwds}{\bm{B}^{>}}
\newcommand{\Predlocbwd}{\bm{b}^{\geq}}
\newcommand{\Predlocbwds}{\bm{b}^{>}}

\begin{document}

\tableofcontents


\title{Constraint LTL Satisfiability Checking without Automata\footnote{This research was partially supported by Programme IDEAS-ERC and Project 227977-SMScom.}}

\author[1]{Marcello M. Bersani}
\author[1]{Achille Frigeri}
\author[1]{Angelo Morzenti}
\author[1]{Matteo Pradella}
\author[1]{Matteo Rossi}
\author[1]{Pierluigi San Pietro}
\affil[1]{Politecnico di Milano - DEIB}

\maketitle


\begin{abstract}
This paper introduces a novel technique to decide the satisfiability of formulae written in the language of
Linear Temporal Logic with both future and past operators and atomic formulae belonging to constraint system $\D$ (CLTLB($\D$) for short).
The technique is based on the concept of \emph{bounded satisfiability}, and hinges on an encoding
of CLTLB($\D$) formulae into QF-EU$\D$, the theory of quantifier-free equality and
uninterpreted functions combined with $\D$.
Similarly to standard LTL, where bounded model-checking and SAT-solvers can be
used as an alternative to automata-theoretic approaches to
model-checking, our approach allows users to solve the
satisfiability problem for CLTLB($\D$) formulae through SMT-solving techniques,
rather than by checking the emptiness of the language of a suitable automaton.
The technique is effective, and it has been implemented in our \zot{} formal verification tool.
\end{abstract}

%
%


\section{Introduction}

Finite-state system verification has attained great successes, both using automata-based and logic-based techniques. 
Examples of the former are the so-called explicit-state model checkers~\cite{Holzmann97} and symbolic model checkers~\cite{CMCH96}. However, some of the best results in practice
have been obtained by logic-based techniques, such as Bounded Model Checking (BMC)~\cite{BCCZ99}.
In BMC, a finite-state machine $A$ (typically, a version of B\"uchi automata) and a desired property $P$ expressed in Propositional Linear Temporal Logic (PLTL) 
are translated into a Boolean formula  $\phi$ 
to be fed 
to a SAT solver. The translation is made finite by bounding the number of time instants.
However,   infinite behaviors, which are crucial in proving, e.g., liveness properties, are also
considered by using the well-known property that a B\"uchi automaton accepts an infinite behavior
if, and only if, it accepts an infinite periodic behavior. 
Hence, chosen a bound $k>0$, a Boolean formula $\phi_k$ is built, such that $\phi_k$ is satisfiable if and
only if there exists an infinite periodic behavior 
of the form $\alpha \beta^\omega$, with $|\alpha \beta| \le k$, that is compatible with system $A$ while violating property $P$. 
This procedure allows counterexample detection, but it is not complete, since the violations of property $P$ requiring ``longer`` behaviors,
i.e., of the form $\alpha \beta^\omega$ with $|\alpha \beta| > k$, are not detected. 
However, in many practical cases it is possible to find bounds large enough for representing counterexamples, but small enough so that the SAT solver can actually find them in a reasonable 
time. 

Clearly, the BMC procedure can be used to check satisfiability of a PLTL formula, without
considering a finite state system $A$.  This has practical applications, since a PLTL formula can represent both the system
and the property to be checked (see, e.g., \cite{PMS12}, where 
the translation into Boolean formulae is made more specific for dealing with satisfiability checking
and metric temporal operators). 
We call this case {\em Bounded Satisfiability Checking} (BSC), which consists
in solving a so-called Bounded Satisfiability Problem:
Given a PLTL formula $P$, and chosen a bound $k>0$, define a Boolean formula $\phi_k$ such
that $\phi_k$ is satisfiable if, and only if, there exists an infinite periodic behavior 
of the form $\alpha \beta^\omega$, with $|\alpha \beta| \le k$, that satisfies $P$.

More recently, great attention has been given to the automated verification of {\em infinite}-state systems.  
In particular, many extensions of temporal logic and automata have been proposed, typically by adding integer variables and arithmetic constraints.
For instance, 
PLTL has been extended to allow formulae with various kinds of arithmetic
constraints \cite{CC00,DD02}. This has led to the study of CLTL($\D$), a general framework extending the future-only fragment of PLTL
by allowing arithmetic constraints belonging to a generic constraint system $\D$. 
The resulting logics are expressive and well-suited to define infinite-state systems and their properties, 
but, even for the bounded case, their satisfiability is typically undecidable \cite{DG06}, since
they can simulate general two-counter machines when $\D$ is powerful enough (e.g., Difference Logic).

However, there are some decidability results, which allow in principle for some kind of automatic
verification.
Most notably,  satisfiability of CLTL($\D$)
is decidable (in $\pspace$) when $\D$ is the class of Integer Periodic Constraints (IPC$^*$) \cite{DG07}, or when 
it is the 
structure $\triple{D}{<}{=}$ with $D\in\{\N,
\Zed, \Q, \Real\}$ \cite{DD07}.
In these cases, decidability is shown by using an automata-based approach similar to the standard case for LTL, by reducing satisfiability checking to the verification of the emptiness of B\"uchi automata.
Given a CLTL($\D$) formula $\phi$, with $\D$ as in the above cases, it is possible to define an automaton $\A_\phi$ such that $\phi$ is satisfiable   
if, and only if, the language recognized by $\A_\phi$ is not empty. 

These results, although of great theoretical interest, are of limited practical relevance for what concerns a possible implementation, since the involved constructions are very inefficient, as they rely on the complementation of B\"{u}chi automata.

In this paper, we extend the above results to a more general logic, called CLTLB($\D$),
which is an extension of PLTLB (PLTL with Both future and past
operators) with arithmetic constraints in constraint system $\D$, and define a
procedure for satisfiability checking that does not rely on automata constructions. 

The idea of the procedure is to determine satisfiability 
by checking a finite number of $k$-satisfiability problems.
Informally, $k$-satisfiability amounts to looking for ultimately
periodic {\em symbolic}  models of the form $\alpha\beta^\omega$, i.e., such that
prefix $\alpha\beta$ of length $k$ admits a bounded arithmetic model (up to instant $k$).
Although the $k$-bounded problem is defined with respect to a bounded
arithmetical model, it provides a representation of
infinite symbolic models by means of ultimately periodic words.
When CLTLB($\D$) has the property that its ultimately periodic symbolic
models, of the form $\alpha\beta^\omega$, always admit an arithmetic model, 
then the $k$-satisfiability problem can be reduced to
satisfiability of QF-EU$\D$ (the theory of quantifier-free
equality and uninterpreted functions combined with $\D$). In this case, $k$-satisfiability is equivalent to satisfiability over infinite models.

There are important examples of constraint systems $\D$, such as for example IPC$^*$, in which determining the existence of arithmetical models is achieved by complementing a B\"{u}chi automaton $\A_C$.
In this paper we define a novel condition, tailored to ultimately periodic models of the form $\alpha\beta^\omega$, which is proved to be equivalent to the one captured by automaton $\A_C$.
Thanks to this condition, checking for the existence of arithmetical models can be done in a bounded way, without resorting to the construction (and the complementation) of B\"{u}chi automata.
This is the key result that makes our decision procedure applicable in practice.

Symmetrically to standard LTL, where bounded model-checking and
SAT-solvers can be
used as an alternative to automata-theoretic approaches to
model-checking, reducing satisfiability to
$k$-satisfiability allows us to determine the satisfiability of CLTLB($\D$) formulae through Satisfiability Modulo Theories (SMT) solvers, instead of checking the emptiness of a B\"uchi automaton.
Moreover, when the length of all prefixes
$\alpha\beta$ to be tested is bounded by some $K \in \Nat$, 
then the number of bounded problems to be solved is finite.
Therefore, we also
prove that $k$-satisfiability is {\em complete} with respect to
the satisfiability problem, i.e., by checking at most $K$ 
bounded problems the satisfiability of CLTLB($\D$) formulae can always be determined.

To the best of our knowledge, our results provide the first effective implementation of a procedure for solving the CLTLB($\D$) satisfiability problem: we show that the encoding into QF-EU$\D$ is linear in the size of the formula to be checked and quadratic in the length $k$. The procedure is implemented in the \zot{} toolkit\footnote{\url{http://zot.googlecode.com}}, which relies on standard SMT-solvers, such as Z3 \cite{z3}.

The paper is organized as follows. 
Section~\ref{sec:cltlb} describes CLTL($\D$) and CLTLB($\D$), and their main known decidability results and techniques.
Section~\ref{sec:SATnoaut} defines the $k$-satisfiability problem, introduces the bounded encoding of CLTLB($\D$) formulae, and shows its correctness.
Section~\ref{subsection-A_C} introduces a novel, bounded condition for checking the satisfiability of CLTLB($\D$) formulae when $\D$ is IPC$^*$, and discusses some cases under which the encoding can be simplified.
Section~\ref{sec:ComplAndCompl} studies the complexity of the defined encoding and proves that, provided that $\D$ satisfies suitable conditions, there exists a completeness threshold.
Section~\ref{sec-applications} illustrates an application of the CLTLB logic and the \zot{} toolkit to specify and verify a system behavior.
Section~\ref{section-related-works} describes relevant related works.
Finally, Section~\ref{section-conclusions} concludes the paper highlighting some possible applications of the implemented decision procedure for CLTLB($\D$). 

\section{Preliminaries}\label{sec:cltlb}

This section presents an extension to Kamp's \cite{Kam68} PLTLB,
by allowing formulae over a constraint system.
As suggested in \cite{CC00}, and unlike the approach of \cite{D04}, the propositional variables
of this logic are 
Boolean terms or atomic arithmetic constraints.

\subsection{Language of constraints}

Let $V$ be a finite set of variables; a {\em constraint system} is a pair
$\D=(D, \mathcal{R})$ where $D$ is a specific domain of interpretation for variables and
constants and
$\mathcal{R}$ is a family of relations on $D$.
An {\em atomic $\D$-constraint} is a term of the form
$R(x_1,\dots,x_n)$, where $R$
is an $n$-ary relation of $\mathcal{R}$ on domain $D$ and $x_1, \dots, x_n$ are variables.
A $\D$-valuation is a mapping $v: V \to D$, i.e., an assignment of a value in $D$
to each variable.
A constraint is {\em satisfied} by a $D$-valuation $v$,
written $v \models_\D R(x_1,\dots,x_n)$, if $( v(x_1),\dots,v(x_n) )
\in R$.

In Section \ref{subsection-A_C} we consider $\D$ to be
Integer Periodic Constraints (IPC$^*$) or its fragments (e.g.,
$\triple{\Zed}{<}{=}$ or $\triple{\N}{<}{=}$) and
$\triple{D}{<}{=}$ when $<$ is a dense order without endpoints, e.g.,
  $D \in \{\mathbb{R}, \Q\}$. 
The language IPC$^*$ is defined by the following
grammar, where $\xi$ is the axiom:
\begin{gather*}
\xi := \theta \mid x < y \mid \xi \wedge \xi \mid \neg \xi \\
\theta := x \equiv_c d \mid x \equiv_c y + d \mid x = y \mid x <
d\mid x=d \mid \theta \wedge \theta \mid \neg \theta
\end{gather*}
where $x,y \in V$, $c \in \N^+$ and $d \in \Zed$.
The first definition of IPC$^*$ can be found in \cite{DG05}; it
is different from ours since it allows existentially
quantified formulae (i.e., $\theta := \exists x \: \theta$) to be part of
the language.
However, since IPC$^*$ is a fragment of Presburger arithmetic, it
has the same expressivity as the above quantifier-free version (but with an exponential blow-up to remove quantifiers).

Given a valuation $v$, the satisfaction relation $\models_\D$ is defined:
\begin{itemize}
\item $v \models_\D x \sim y$ iff $v(x) \sim v(y)$;
\item $v \models_\D x \sim d$ iff $v(x) \sim d$;
\item $v \models_\D x \equiv_c d$ iff $v(x) - d = kc$ for some $k \in \Zed$;
\item $v \models_\D x  \equiv_c y + d$ iff $v(x) - v(y) - d = kc$ for some
  $k \in \Zed$;
\item $v \models_\D \xi_1 \wedge \xi_2$ iff $v \models_\D \xi_1$ and
  $v \models_\D \xi_2$;
\item $v \models_\D \neg\xi$ iff $v \not\models_\D \xi$;
\end{itemize}
where $\sim$ is either $=$ or $<$.
A constraint is \emph{satisfiable} if there exists a valuation $v$ such
that $v \models_\D \xi$.
Given a set of IPC$^*$ constraints $C$, we write $v \models_\D C$
when $v \models_\D \xi$ for every $\xi \in C$.

\subsection{Syntax of CLTLB}\label{subsection-syntax}

CLTLB($\D$) is defined as an extension of PLTLB,
where atomic formulae are relations from $\mathcal{R}$ over arithmetic temporal terms defined in $\D$.
The resulting logic is actually equivalent to the quantifier-free fragment of
first-order LTL over signature $\mathcal{R}$.
Let $x$ be a variable; \emph{arithmetic temporal terms} (a.t.t.) are
defined as:
\[
   \alpha := c \mid x \mid \aX \alpha \mid \aY \alpha.
\]
where $c$ is a constant in $D$ and $x$ is a variable over $D$.
The syntax of (well formed) formulae of CLTLB($\D$) is recursively defined as follows:
\begin{equation*}
  \phi :=
  \begin{gathered}
    R(\alpha_1, \dots, \alpha_n) \mid \phi \wedge \phi \mid \neg \phi \mid
   \X\phi \mid \Y\phi 
\mid \phi\U\phi \mid \phi\Snc\phi
  \end{gathered}
\end{equation*}
where $\alpha_i$'s are a.t.t.'s, $R \in \mathcal{R}$; $\X$, $\Y$, $\U$, and $\Snc$  are the
usual ``next'',  ``previous'', ``until'',  and ``since'' 
operators from LTL.

Note that $\aX$ and $\X$ are two distinct operators;
if $\phi$ is a formula, $\X \phi$ has the standard PLTL meaning,
while $\aX \alpha$ denotes the \emph{value}
of a.t.t. $\alpha$ in the next time instant.
The same holds for $\aY$ and $\Y$, which refer to the previous
time instant.
Thanks to the obvious property that
$\aX \aY x \equiv \aY \aX x \equiv x$,  
we will assume, with no loss of generality, that a.t.t.'s do not
contain any nested alternated occurrences of the operators $\aX$ and $\aY$. 
Each relation symbol is associated with a natural number denoting
its arity.
As we will see in Section \ref{section-correctness}, we can treat
separately $0$-ary relations, i.e., propositional letters, whose set is denoted by
$\mathcal{R}_0$. We also write  CLTLB$(\D,\mathcal{R}_0)$ to denote
the language CLTLB over 
the constraint system $\D$ whose $0$-ary relations are exactly those in $\mathcal{R}_0$.
CLTL$(\D)$ is the future-only fragment of CLTLB$(\D)$.

The \emph{depth} $|\alpha|$ of an a.t.t. is
the total amount of temporal shift needed in evaluating $\alpha$: 
\begin{equation*}
|x| = 0, \quad |\aX\alpha| = |\alpha| + 1, \quad |\aY\alpha| = |\alpha| - 1.
\end{equation*}

Let $\phi$ be a CLTLB($\D,\mathcal{R}_0$) formula, $x$ a variable of $V$ and
$\Gamma_x(\phi)$ the set of all a.t.t.'s occurring in $\phi$ in which $x$ appears. We define the ``look-forwards''
$\lceil \phi \rceil_x$ and
``look-backwards''
$\lfloor \phi \rfloor_x$ of $\phi$ relatively to $x$ as:
\[
\lceil \phi \rceil_x = \max_{{\alpha_i} \in \Gamma_x(\phi)}\{0,
|\alpha_i|\}, \quad
\lfloor \phi \rfloor_x = \min_{{\alpha_i} \in \Gamma_x(\phi)}\{0, |\alpha_i|\}.
\]
The definitions above naturally extend to $V$ by letting
$\lceil \phi \rceil = \max_{x\in
V} \{\lceil \phi \rceil_x\}$,
$\lfloor \phi \rfloor =  \min_{x\in V}\{\lfloor \phi \rfloor_x\}$.
Hence, $\lceil \phi \rceil$ ($\lfloor \phi
\rfloor$) is the largest (smallest) depth of all the a.t.t.'s of $\phi$, representing
the length of the future (past) segment needed to evaluate 
$\phi$ in the current instant.

\subsection{Semantics}\label{subsection-semantics}

The semantics of CLTLB($\D,\mathcal{R}_0$) formulae is defined with respect to
a strict linear order representing time $\pair{\Zed}{<}$.
Truth values of propositions in $\mathcal{R}_0$, and values of variables belonging
to $V$ are defined by a pair
$\pair{\pi}{\sigma}$ where $\sigma : \Zed \times V \to D$ is a function
which defines the value of variables at each position in $\Zed$ and
$\pi :\Zed \to \wp(\mathcal{R}_0)$ is a function associating a subset of the set of
propositions with each element of $\Zed$.
Function $\sigma$ is extended to terms as follows:
\[
\sigma(i, \alpha) = \sigma(i+|\alpha|, x_{\alpha})
\]
where $x_\alpha$ is the variable in $V$ occurring in term $\alpha$, if any; otherwise $x_\alpha = \alpha$.
The semantics of a CLTLB($\D,\mathcal{R}_0$) formula $\phi$ at instant $i\geq 0$
over a linear structure $\pair{\pi}{\sigma}$ is recursively defined by
means of a satisfaction relation
$\models$ as follows, for every formulae $\phi, \psi$ and
for every a.t.t. $\alpha$:
\begin{equation*}
\begin{aligned}
\pair{\pi}{\sigma}, i \models p  &\iFF  p \in \pi(i) \text{ for } p \in \mathcal{R}_0 \\
\pair{\pi}{\sigma}, i \models R(\alpha_1, \dots, \alpha_n) &\iFF
  (\sigma(i+|\alpha_1|, x_{\alpha_1}), \dots, \sigma(i+|\alpha_n|, x_{\alpha_n}) ) \in R
  \text{ for } R\in\mathcal{R}\setminus\mathcal{R}_0 \\
\pair{\pi}{\sigma}, i \models \neg \phi &\iFF  \pair{\pi}{\sigma},i \not\models \phi \\
\pair{\pi}{\sigma}, i \models \phi \wedge \psi &\iFF  \pair{\pi}{\sigma}, i \models \phi
\, \text{and} \, \pair{\pi}{\sigma}, i \models \psi\\
\pair{\pi}{\sigma}, i \models \X \phi &\iFF \pair{\pi}{\sigma},i+1 \models \phi \\
\pair{\pi}{\sigma}, i \models \Y \phi &\iFF \pair{\pi}{\sigma},i-1 \models
\phi \wedge i>0\\
\pair{\pi}{\sigma}, i \models \phi\U\psi &\iFF
\exists \, j\geq i: \pair{\pi}{\sigma},j \models \psi \ \wedge
\pair{\pi}{\sigma},n \models \phi \ \forall n: \ i\leq n < j \\
\pair{\pi}{\sigma}, i \models \phi\Snc\psi &\iFF
\exists \, 0\leq j \leq i: \pair{\pi}{\sigma},j \models \psi \, \wedge
\pair{\pi}{\sigma},n \models \phi \ \forall n: \ j < n \leq i. \\
\end{aligned}
\end{equation*}

A formula $\phi \in$ CLTLB($\D,\mathcal{R}_0$) is \emph{satisfiable} if there exists
a pair $\pair{\pi}{\sigma}$ such that $\pair{\pi}{\sigma},0 \models
\phi$; in this case, we say that $\pair{\pi}{\sigma}$ is a \emph{model} of $\phi$, $\pi$
is a \emph{propositional model} and $\sigma$ is an \emph{arithmetic model}.
By introducing as primitive the connective $\vee$, the dual operators
``release'' $\R$, ``trigger'' $\T$ and ``previous'' $\Z$
are defined as: 
$\phi\R \psi
\equiv \neg(\neg \phi \U \neg\psi)$, $\phi\T \psi
\equiv \neg(\neg \phi \Snc \neg\psi)$ and $\Z \phi \equiv \neg \Y
\neg \phi$;
by applying De Morgan's rules,
we may assume every CLTLB formula to be in {\em positive normal form}, i.e.,
negation may only occur in front of atomic propositions and relations.

\subsection{CLTL with automata}
\label{section-symbolic}

The {\em satisfiability} problem for CLTL formula $\phi$ consists in determining whether there exists a model $\pair{\pi}{\sigma}$ for $\phi$ such that  $\pair{\pi}{\sigma},0 \models \phi$.
In this section, we recall some known results where the propositional part $\pi$ of $\pair{\pi}{\sigma}$ is either missing or can be eliminated (hence, with a slight abuse of notation we will write $\sigma,0 \models \phi$ instead of $\pair{\pi}{\sigma},0 \models \phi$).

Hereafter, we restrict $\D$ to be the structure defined by
$\IPC^*$, or by $\triple{D}{<}{=}$, where $D \in \set{\Nat, \Zed, \Q,\Real}$.
For such constraint systems a decision procedure based on B\"uchi automata is studied in \cite{DD07}.
The presented notions are essential to develop our decision procedure without automata construction.

Let $\phi$ be a  CLTLB($\D$) formula and $\terms(\phi)$ be the set of
arithmetic terms of the form $\aX^i x$ for all $0\leq i\leq \lceil
\phi\rceil$ or of the form $\aY^i x$ for all $1\leq i\leq - \lfloor
\phi\rfloor$ and for all $x \in V$.
If domain $D$ is discrete, let $const'(\phi)=\set{m, \dots, M}$ be the set of
constants occurring in $\phi$, where $m, M \in D$ are the minimum and
maximum constants.
We extend $const'(\phi)$ to the set $\const(\phi) = \interval{m}{M}$ of all values 
between $m$ and $M$.

A set of  $\D$-constraints over $\terms(\phi)$ is {\em maximally consistent} if for every
$\D$-constraint $\theta$ over $\terms(\phi) \cup \const(\phi)$, either $\theta$ or
$\neg\theta$ is in the set.

\begin{definition}\label{def-sv}
A \emph{symbolic valuation} $sv$ for $\phi$ is a maximally consistent
set of $\D$-constraints over $\terms(\phi)$ and $\const(\phi)$.
\end{definition}

The original definition of symbolic valuation for IPC$^*$ constraint
systems in \cite{DG05} is slightly different.
Our definition does not consider explicitly relation $x=d$ and periodic relation $x\equiv_c d$, with $c,d \in D$,
because they  are inherently represented in the $k$-bounded arithmetical models defined in Section
\ref{section-bsp}.
Equality between variables and constants do not require to be symbolically represented by a symbolic constraint of the form $x=d$ as  $k$-bounded arithmetical models associate each variable with an ``explicit'' value from $D$.
Moreover, given $x$ a value from $D$, relation $x\equiv_c d$ is inherently defined.

The satisfiability of a symbolic valuation is defined as follows, by considering each a.t.t. as a new fresh variable.

\begin{definition}
\label{def:fresh}
The set of all symbolic
valuations for $\phi$ is denoted by $SV(\phi)$.
Let $A$ be a set of variables and $\fresh: \terms(\phi) \rightarrow A$ be an injective function mapping each
a.t.t of $\phi$ to a fresh variable in set $A$. Function $\fresh$ is naturally
extended to every symbolic valuation $sv$ for $\phi$, by replacing
each a.t.t. $\alpha \in \terms(\phi)$ in $sv$ with $\fresh(\alpha)$.
Symbolic valuations for $\phi$ are now defined over the set $\fresh(\terms(\phi))$.
A symbolic valuation $sv$ for $\phi$ is \emph{satisfiable} if there exists a
$\D$-valuation $v': A \rightarrow D$,
such that $v' \models_ \D \fresh(sv)$, i.e., satisfiability of $sv$ considers all
a.t.t.'s as fresh variables.
\end{definition}

Given a symbolic valuation $sv$ and a $\D$-constraint $\xi$ over a.t.t.'s, we write $sv
\symodels \xi$ if for every $\D$-valuation $v'$ such that $v' \models_\D \fresh(sv)$ then we have $v' \models_\D \fresh(\xi)$.
We assume that the problem of checking $sv \symodels \xi$ is decidable.
The satisfaction relation $\symodels$ can also be extended  to
infinite sequences $\rho: \N \rightarrow SV(\phi)$ (or, equivalently, $\rho \in SV(\phi)^\omega$) of symbolic
valuations; it is the same as $\models$ for all temporal operators except for atomic formulae:
\[
\rho, i \symodels\xi \siff \rho(i) \symodels \xi.
\]
Then, given a CLTLB($\D$) formula $\phi$, we say that a symbolic model $\rho$ \emph{symbolically
satisfies} $\phi$ (or $\rho$ is a \emph{symbolic model} for $\phi$) when $\rho,0 \symodels \phi$.

In the rest of this section we consider CLTLB($\D$) formulae that do not include arithmetic temporal operator $\aY$.
This is without loss of generality, as Property \ref{theorem-removeY} will show.

\begin{definition}
A pair of symbolic valuations $\pair{sv_1}{sv_2}$ for $\phi$ is
\emph{locally consistent} if, for all $R$ in $\D$:
%
\[
R(\aX^{i_1} x_1, \dots, \aX^{i_n} x_n) \in sv_1 \text{ implies }
R(\aX^{i_1-1} x_1, \dots, \aX^{i_n-1} x_n) \in sv_2
\]

with $i_j\geq 1$ for all $j \in [1, n]$.
A sequence of symbolic valuations $sv_0 sv_1 \dots$ is \emph{locally consistent} if
all pairs $\pair{sv_i}{sv_{i+1}}$, $i\ge 0$, are locally consistent.
\end{definition}
A locally consistent infinite sequence $\rho$ of symbolic valuations
\emph{admits an arithmetic model},
if there exists a $\D$-valuation sequence $\sigma$
such that $\sigma, i \models \rho(i)$, for all $i\geq 0$.
In this case, we write $\sigma,0 \models \rho$. 

The following fundamental proposition draws a link between
the satisfiability by sequences of symbolic valuations and by sequences of $\D$-valuations.
\begin{proposition}[\cite{DD07}]\label{prop-arith-model}
A CLTL($\D$) formula $\phi$ is satisfiable if, and only if, there
exists a symbolic model for $\phi$ which admits an arithmetical model,
i.e., there exist $\rho$ and $\sigma$ such that $\rho,0\symodels\phi$
and $\sigma,0\models\rho$. 
\end{proposition}

Following \cite{DD07}, for constraint systems of the form $(D,<,=)$,
where $<$ is a strict total ordering on $D$, it is possible to
represent a symbolic valuation $sv$ by its labeled directed graph
$G_{sv}=(V,\tau)$, $\tau \subseteq V\times\{<,=\}\times V$, such that $(x,\sim,
y)\in \tau$ if, and only if, $x\sim y\in sv$. 
This construction extends also to sequences: given a sequence $\rho$
of symbolic valuations, it is possible to represent $\rho$ via the
graph $G_\rho$ obtained by superimposition of the graphs corresponding
to the symbolic evaluations $\rho(i)$. More formally $G_\rho=(V\times
\N,\tau_\rho)$, where $((x,i),\sim,(y,j))\in \tau_\rho$ if, and only if, either $i\leq j$ and $(x\sim \aX^{j-i}y)\in \rho(i)$, or $i>j$ and $(\aX^{i-j} x\sim y)\in \rho(j)$.

An infinite path $d:\N\rightarrow V\times \N$ in $G_\rho$, is called a \emph{forward} (resp. \emph{backward}) path if:
\begin{enumerate}
\item for all $i\in \N$, there is an edge from $d(i)$ to $d(i + 1)$ (resp., an edge from $d(i + 1)$ to $d(i)$);
\item for all $i\in \N$, if $d(i)=(x,j)$ and $d(i+1)=(x',j')$, then
  $j\leq j'$.
\end{enumerate}
A forward (resp. backward) path is \emph{strict} if there exist infinitely many $i$ for which there is a <-labeled edge from $d(i)$ to $d(i+1)$ (resp., from $d(i+1)$ to $d(i)$).
Intuitively, a (strict) forward path represents a sequence of (strict) monotonic
increasing values whereas a (strict) backward path represents a
sequence of (strict) monotonic decreasing values.

Given a CLTL($\D$) formula $\phi$, it is possible \cite{DD07} to define a B\"{u}chi automaton $\A_\phi$ recognizing symbolic models of $\phi$, and then reducing the satisfiability of $\phi$ to the emptiness of $\A_\phi$.
The idea is that automaton $\A_\phi$ should accept the intersection of the following languages, which defines exactly the language of symbolic models of $\phi$:
\begin{itemize}
\item[(1)] the language of  LTL models $\rho$;
\item[(2)] the language of sequences of locally
  consistent symbolic valuations;
\item[(3)] the language of sequences of
  symbolic valuations which admit an arithmetic model.
\end{itemize}

Language (1) is accepted by the Vardi-Wolper automaton $\A_s$ of
$\phi$ \cite{vw}, while language (2) is recognized by the automaton
$\A_\ell=(SV(\phi),sv_0,\step{}, SV(\phi))$, where
the states are $SV(\phi)$, all accepting; $sv_0$ is the initial state; and the transition relation is such that
$sv_i\step{sv_{i}}sv_{i+1}$ if, and only if,
all pairs $(sv_{i},sv_{i+1})$ are locally consistent \cite{DD07}.

If the constraint system we are considering has the \emph{completion property} (defined next), then all sequences of locally
consistent symbolic valuations admit an arithmetic model, and
condition (3) reduces to (2).

\subsubsection{Completion property}
\label{subsec:completion}
Each automaton involved in the definition of
$\A_\phi$ has the function of \lq\lq filtering\rq\rq\ sequences of symbolic valuations
so that 1) they are locally consistent, 2) they satisfy an LTL
property and 3) they admit a (arithmetic) model.
For some constraint systems, admitting a model is
a consequence of local consistency.
A set of relations over $D$ has the \emph{completion} property if, given:
\begin{description}
\item[(i)] a symbolic valuation $sv$ over a finite set of variables $H
  \subseteq V$,
\item[(ii)] a subset $H' \subseteq H$,
\item[(iii)] a valuation $v'$ over $H'$ such that $v' \models sv'$,
  where $sv'$ is the subset of atomic formulae in $sv$ which uses only
  variables in $H'$
\end{description}
then there exists a valuation $v$ over $V$ extending $v'$ such
that $v \models sv$.
An example of such a relational structure is
$\triple{\mathbb{R}}{<}{=}$.
Let $\triple{D}{<}{=}$ be a relational structure defining the language of
atomic formulae.
We say that $D$ is \emph{dense}, with respect to the order $<$, if for
each $d,d'\in D$ such that $d<d'$, there exists $d''\in D$ such that
$d<d''<d'$,
whereas $D$ is said to be \emph{open} when for each $d \in D$, there
exist two elements $d',d'' \in D$ such that $d'<d<d''$.
\begin{lemma}[Lemma 5.3, \cite{DD07}]
\label{lem:completion}
Let $\triple{D}{<}{=}$ be a relational structure where $D$ is
infinite and $<$ is a total order. Then, it satisfies the completion
property if, and only if, domain $D$ is dense and open.
\end{lemma}
The following result relies on the fact that every
locally consistent sequence of symbolic valuations with respect to
the relational structure $\D$ admits a model.
\begin{proposition}\label{prop-completion}
Let $\D$ be a relational structure satisfying the completion property and
$\phi$ be a CLTL($\D$) formula. Then, the language of sequences of
symbolic valuations which admit a model is $\omega$-regular.
\end{proposition}
In this case the automaton $\A_\phi$ that recognizes exactly all the sequences of
symbolic valuations which are symbolic models of $\phi$ is defined by
the intersection (\emph{\`a la} B\"uchi) $\A_\phi = \A_s \cap \A_\ell$.

In general, however, language (3) may {\em not} be $\omega$-regular. Nevertheless, if the constraint system is of the form $(D,<,=)$, it is possible to define
an automaton $\A_C$ that accepts a superset of language (3), but such that
all its {\em ultimately periodic} words are sequences of symbolic
valuations that admit an arithmetic model. Actually, $\A_C$ recognizes
a sequence $\rho$ of symbolic valuations that satisfies the following property:

\begin{subproperty}\label{property-C}
There do not exist vertices $u$ and $v$ in the same symbolic valuation in $G_\rho$ satisfying 
all the following conditions:
\begin{enumerate}
  \item there is an infinite forward path $d$ from $u$;
  \item there is an infinite backward path $e$ from $v$;
  \item $d$ or $e$ are strict;
  \item for each $i, j \in \N$, whenever $d(i)$ and $e(j)$ belong to the same symbolic valuation, there exists an edge, labeled by $<$, from $d(i)$ to $e(j)$.
\end{enumerate}
\end{subproperty}
Informally, Property~\ref{property-C} guarantees that in the model there does not exist an
infinite forward path whose values are
infinitely often less than values of an infinite backward path; in other
words, an infinite strict/non-strict monotonic increasing sequence of values can not be
infinitely often less than an infinite non-strict/strict monotonic decreasing sequence of values.

The proposed method is general and it can be used whenever
it is possible to build an automaton $\A_C$ which defines
a condition $C$ guaranteeing the existence of a sequence $\sigma$ such that $\sigma,0 \models \rho$.
In particular, for constraint systems IPC$^*$, $(\N, <, =)$, and $(\Zed, <, =)$, $\A_C$ can be effectively built.
Let $\A_\phi$ be defined as the (B\"uchi) product of $\A_\ell, \A_s, \A_C$; since emptiness of B\"{u}chi automata can be checked just on ultimately periodic words,
the language of $A_\phi$ is empty if, and only if, $\phi$ does not have a symbolic model (which is equivalent to not having an arithmetical model).

When the condition $C$ is sufficient and necessary for the existence
of models $\sigma$ such that $\sigma,0 \models \rho$, then automaton
$\A_\phi$ represents all the sequences of symbolic valuations which admit
a model $\sigma$.
A fundamental lemma, on which Proposition \ref{prop-sat-CLTL} below relies, draws a sufficient and necessary condition for the existence of models of sequences of symbolic valuations.
\begin{lemma}[\cite{DD07}]\label{lemma-condC}
Let $\rho$ be an ultimately periodic sequence of symbolic valuations of the form $\rho= \alpha \beta^\omega \in SV(\phi)^\omega$ that is locally consistent.
Then, $\sigma,0 \models \rho$ (i.e., $\rho$ admits a model $\sigma$) if, and only if, $\rho$ satisfies $C$.
\end{lemma}

Therefore, the satisfiability problem can be solved by checking the
emptiness of the language recognized by the automaton $\A_\phi$.

\begin{proposition}[\cite{DD07}]\label{prop-sat-CLTL}
A CLTL($\D$) formula $\phi$ is satisfiable if, and only if, the language $\Lng(\A_\phi)$ is not empty.
\end{proposition}

In the next section, we provide a way for checking the satisfiability of CLTLB($\D$) formulae that does not require the construction of automata $\A_s$, $\A_\ell$ and $\A_C$.
Our approach takes advantage of the semantics of CLTLB($\D$) for building models of formulae through a semi-symbolic construction.
We use a reduction to a Satisfiability Modulo Theories (SMT) problem which extends the one proposed for Bounded Model Checking \cite{BCCSZ03}.
In the automata-based construction the definition of automaton $\A_\phi$ may be prohibitive in practice and requires to devise alternative ways that avoid the exhaustive enumeration of all the states in $\A_\phi$.
In fact, the size of $\A_s$ is exponential with respect to the size of the formula and condition $C$, which needs to be checked when the constraint system does not have the completion property, as in the case of $\triple{\Zed}{<}{=}$, is defined by complementing, through Safra's algorithm, automaton $\A_{\neg C}$ which recognizes symbolic sequences satisfying  the negated condition $C$ \cite{DD07}.
Since to show the satisfiability of a formula one can exhibit an ultimately periodic model whose length may be much smaller than the size of automaton $\A_\phi$, in many cases the whole construction of $\A_\phi$ is useless.
However, proving unsatisfiability is comparable in complexity to defining the whole automaton $\A_\phi$ because it requires to verify that no ultimately periodic model $\alpha\beta^\omega$ can be constructed for size $|\alpha\beta|$ equal to the size of automaton $\A_\phi$.
Motivated by the arguments above, we define the bounded satisfiability problem which
consists in looking for ultimately periodic symbolic models
$\alpha\beta^\omega$ such that prefix $\alpha\beta$ is of fixed
length (which is an input of the problem) and which admits a \emph{finite} arithmetical model
$\sigma_k$.
Since symbolic valuations partition the space of variable valuations, an assignment of values to terms uniquely identifies a symbolic valuation (see next Lemma \ref{lemma-A}).
For this reason, we do not need to precompute the set $SV(\phi)$ and instead we enforce the periodicity between a pair of sets of relations, those defining the first and last symbolic valuations in $\beta$.
We show that, when a formula $\phi$ is boundedly satisfiable, then it is
also satisfiable.
We provide a (polynomial-space) reduction from the bounded satisfiability
problem to the satisfiability of
formulae in the quantifier-free theory of equality and uninterpreted
functions QF-EUF combined with $\D$.

\section{Satisfiability of CLTLB($\D$) without automata}
\label{sec:SATnoaut}

In this section, we introduce our novel technique to solve the satisfiability problem of CLTLB($\D$) formulae without resorting to an automata-theoretic construction.

First, we provide the definition of the $k$-satisfiability problem 
for CLTLB$(\D)$ formulae in terms of the existence of a so-called $k$-bounded 
arithmetical model
$\sigma_k$, which provides a finite representation of
infinite symbolic models by means of ultimately periodic words.
This allows us to prove that $k$-satisfiability is still
representative of the satisfiability problem as defined in Section
\ref{subsection-semantics}. 
In fact, for some constraint systems, a bounded solution can be used
to build the infinite model $\sigma$ for the formula from the $k$-bounded one
$\sigma_k$ and from its symbolic model.
We show in Section \ref{section-correctness} that a formula $\phi$ is satisfiable if, and only if, it is $k$-satisfiable and its bounded solution $\sigma_k$ can be used to derive its
infinite model $\sigma$.
In case of negative answer to a $k$-bounded instance, we can not
immediately entail the unsatisfiability of the formula. 
However, we prove in Section \ref{section-completeness} that for every formula $\phi$ there exists an upper bound $K$,  which 
can  effectively be determined, such that if $\phi$ is not $k$-satisfiable for all $k$ in $\interval{1}{K}$ then 
$\phi$ is unsatisfiable.

\subsection{Bounded Satisfiability Problem}
\label{section-bsp}

We first define the Bounded Satisfiability Problem (BSP), by considering
bounded symbolic models of CLTLB($\D$) formulae.
For simplicity, we consider the set $\mathcal{R}_0$ of propositional letters to be empty; this is without loss of generality, as Property \ref{theorem-removeAP} of Section \ref{section-correctness} attests.
A bounded symbolic model is, informally, a finite representation of infinite
CLTLB($\D$) models over the alphabet of
symbolic valuations $SV(\phi)$.
We restrict the analysis to ultimately periodic symbolic models, i.e., of the form
$\rho = \alpha \beta^\omega$.
Without loss of generality, we consider models where $\alpha = \alpha' s$ and $\beta = \beta' s$ for some symbolic valuation $s$.
BSP is defined with respect to a $k$-bounded
model $\sigma_k : \{\lfloor\phi\rfloor, \dots, k+\lceil\phi\rceil\}\times V\rightarrow
D$, a finite sequence $\rho'$ (with $|\rho'|=k+1$)  of symbolic valuations and a $k$-bounded satisfaction relation $\models_k$ defined as follows:
\[
\sigma_k,0 \models_k \rho' \text{ iff } \sigma_k, i \models \rho'(i)
\text{ for all } 0\leq i \leq k.
\]

The {\em $k$-satisfiability problem} of formula $\phi$ is defined as follows:
\begin{description}
\item[Input] A CLTLB($\D$) formula $\phi$, a constant $k\in \N$
\item[Problem] Is there an ultimately periodic sequence of symbolic
  valuations $\rho = \alpha\beta^\omega$ with $|\alpha\beta| = k+1$, $\alpha=\alpha's$ and $\beta=\beta's$, such that:
\begin{itemize}
\item $\rho, 0 \symodels \phi$ and
\item there is a $k$-bounded model $\sigma_k$ for which $\sigma_k,0 \models_k \alpha \beta$?
\end{itemize}
\end{description}
Since $k$ is fixed, the procedure for determining the satisfiability of CLTLB($\D$) formulae over
bounded models is not complete:
even if there is no accepting run of automaton
$\A_\phi$ when $\rho'$ as above has length $k$, there may be accepting runs for a larger $\rho'$.

\begin{definition}\label{completeness}
Given a CLTLB($\D$) formula $\phi$, its \emph{completeness threshold} $K_\phi$, if it exists, is the smallest number such that $\phi$ is satisfiable if and only if $\phi$ is $K_\phi$-satisfiable.

\end{definition}

\subsection{Avoiding explicit symbolic valuations}

The next, fundamental Lemma~\ref{lemma-A} and Lemma~\ref{lemma-B} state how $k$-bounded
models $\sigma_k$ are representative of ultimately periodic sequences
of symbolic valuations, i.e., of symbolic models of the formula.
More precisely, Lemma \ref{lemma-B} allows for building a sequence of
symbolic valuations from $\sigma_k$. 
Hence, the encoding described in the following Section~\ref{section-encoding} can 
consider only atomic
subformulae 
occurring in CLTLB($\mathcal{D}$) formula $\phi$, even though the BSP for $\phi$ is defined with respect to sequences of
symbolic valuations. The encoding also introduces additional constraints,
to enforce periodicity of relations in $\mathcal{R}$, thus allowing us to derive
an ultimately periodic symbolic model from $\sigma_k$.


To exploit this property, we adopt a special requirement on the constraint system.
In fact, Lemma \ref{lemma-A} and Lemma \ref{lemma-B} rely on
the following assumption, which guarantees the uniqueness of the symbolic valuation given an assignment to variables.
\begin{definition} 
A constraint system $\D= (D,\mathcal{R})$ is \textit{value-determined} if, for all $m$ and for all $v \in D^m$, there exists at most one $m$-ary relation $R \in \mathcal{R}$ s.t. $v \models_{\D} R$.
\end{definition}

For value-determined constraint systems we avoid the definition of set $SV(\phi)$ as we are allowed to derive symbolic models for $\phi$ through $\sigma_k$.
Therefore, our approach is general and it can be used to solve CLTLB$(\D)$ for a value-determined constraint system $\D$,
which is the case of the constraint systems presented in Section~\ref{section-symbolic}.

\begin{lemma}\label{lemma-A}
Let $\D= (D,\mathcal{R})$ be a value-determined constraint system, $\phi$
be a CLTLB$(\D)$ formula and $v$ be a $\D$-valuation extended
to terms appearing in symbolic valuations of $SV(\phi)$.
Then, there is a unique symbolic valuation $sv$ such that
$
v \models_{\D} sv.
$
\end{lemma}

\begin{proof}

Let $sv$ be the symbolic valuation, defined from the values in $v$, such that, for any $R \in \mathcal{R}$,
if $v \models_{\D}
\fresh(R(\alpha_1,\dots,\alpha_n))$ then $R(\alpha_1,\dots,\alpha_n) \in sv$ (where $\fresh$ is the mapping introduced in Definition~\ref{def:fresh}
to replace arithmetic temporal terms with fresh variables).
We show that $sv$ is maximally consistent.
Consistency is immediate, since if $v \models_{\D}
\fresh(R(\alpha_1,\dots,\alpha_n))$ then $v \models_{\D}
\fresh(\neg R(\alpha_1,\dots,\alpha_n))$ cannot hold.
By contradiction, assume that $sv$ is not maximal, 
i.e., 
there is a relation $R' \not\in sv$ such that $v \models_{\D} \fresh(R')$, and $sv \cup \set{R'}$ is consistent. 
Hence,
$v \models_{\D}  sv \cup \set{R'}$.
By definition, a symbolic valuation $sv$ includes all relations among the terms of $\phi$, hence there is a relation $R'' \in sv$, with $R''\neq R'$,  over the same set
of terms of $R'$.
Hence, in constraint
system $\D$ we have two different relations, $R'$ and $R''$, over the same set of
terms and such that $v\models_{\D} \fresh(R')$ and $v \models_{\D} \fresh(R'')$. But this contradicts the assumption that $\D$ is value-determined.
\end{proof}

\begin{corollary}\label{corollary-lemma-A}
  Let $\phi$ be a CLTLB$(\D)$ formula, $v$ a $\D$-valuation extended
  to terms of symbolic valuations and $sv$ a symbolic valuation in
  $SV(\phi)$.  Then, for $v \models_{\D} sv$ and for all relations $R
  \in \mathcal{R}$
$$
sv\symodels R(\alpha_1,\dots,\alpha_n) \text{ iff } v \models_{\D} \fresh(R(\alpha_1,\dots,\alpha_n)).
$$
\end{corollary}

\begin{proof}
Suppose that $sv \symodels R(\alpha_1,\dots,\alpha_n)$.
By definition, $sv \symodels
R(\alpha_1,\dots,\alpha_n)$ holds if, for every $\D$-valuation $v'$ (over the
set of terms in $sv$) such that $v' \models_{\D} sv$,  $v'\models_{\D}
\fresh(R(\alpha_1,\dots,\alpha_n))$ holds.
Therefore, also $v \models_{\D} \fresh(R(\alpha_1,\dots,\alpha_n))$.
%
The converse is an immediate consequence of Lemma \ref{lemma-A}.%
%
\end{proof}

\begin{lemma}\label{lemma-B}
Let $\phi$ be a CLTLB$(\D)$ formula and $\sigma_k$
be a finite sequence of $\D$-valuations.
Then, there exists a unique locally consistent sequence $\rho \in SV(\phi)^{k+1}$ such that
$\sigma_k,i \models \rho(i)$, for all $i \in \interval{0}{k}$.
\end{lemma}

\begin{proof}
By Lemma~\ref{lemma-A} it follows that, for all $i \in \interval{0}{k}$,
the assignment of variables defined by $\sigma_k$ is such that
$\sigma_k,i \models \rho(i)$ and $\rho(i)$ is unique.
By Corollary~\ref{corollary-lemma-A}, values in $\sigma_k$ from
position $i$ satisfy a relation
$R$ at position $i$ if, and only if, $R$ belongs to symbolic
valuation $\rho(i)$, i.e., $\rho(i) \symodels R \text{ iff } \sigma_k,i
\models \fresh(R)$.
In addition, any two adjacent symbolic
valuations $\rho(i)$ and $\rho(i+1)$ are locally consistent, i.e., both 
$R(\aX^{i_1}x_1, \dots, \aX^{i_n}x_n) \in
\rho(i)$ and $R(\aX^{i_1-1}x_1, \dots, \aX^{i_n-1}x_n) \in \rho(i+1)$.
In fact, the evaluation in $\sigma_k$  of an arithmetic term $\aX^{i_j}x_j$ in position $i$
is the same as the evaluation of $\aX^{i_j-1}x_j$ in
position $i+1$. 
%
\end{proof}

\subsection{An encoding for BSP without automata}
\label{section-encoding}

We now show how to encode a CLTLB($\D$) formula
into a quantifier-free formula in the theory
$\text{EUF}\cup \mathcal{D}$ (QF-EU$\mathcal{D}$), where EUF is the theory of Equality and Uninterpreted
Functions. 
This is the basis for reducing the 
BSP for CLTLB($\D$) to the satisfiability of  QF-EU$\mathcal{D}$, as proved in Section~\ref{section-correctness}.
Satisfiability of  QF-EU$\mathcal{D}$ is decidable, provided that $\D$ includes a copy of $\N$ with the successor relation and that
$\text{EUF}\cup \mathcal{D}$ is consistent, as in our case. The latter condition is
easily verified in the case of the union of two consistent, disjoint,
stably infinite theories (as is the case for EUF and arithmetic). 
\cite{BCFPR10} describes a similar approach for the case of
Integer Difference Logic (DL) constraints.
It is worth noting that standard LTL can be encoded by a formula in
QF-EU$\mathcal{D}$ with $\D = \pair{\Nat}{<}$, rather than in Boolean logic \cite{BHJLS06}, resulting 
in a more succinct encoding.

The encoding presented below represents ultimately periodic sequences of symbolic
valuations $\rho$ of the form $sv_0 sv_1\dots
sv_{loop-1}(sv_{loop}\dots sv_k)^\omega$.
To do this, we use a positive integer variable $loop$ for which we
require $sv_{loop-1}=sv_k$.
Therefore, we look for a finite word $\rho'=sv_0 sv_1\dots
sv_{loop-1}(sv_{loop}\dots sv_k) sv_{loop}$ of length $k+2$ representing the
ultimately periodic model above.
Instant $k+1$ in the encoding is used to correctly represent
the periodicity of $\rho$ by constraining atomic formulae (propositions
and relations) at positions $loop$ and $k+1$.
Moreover, all subformulae of $\phi$ that hold at position $loop-1$ must also hold in $k$.

\paragraph{Encoding terms}\label{subsection-encoding-terms}

We introduce \textit{arithmetic formula functions} to encode the
terms in set $\terms(\phi)$. 
Let $\alpha$ be a term in $\terms(\phi)$, then the arithmetic
formula function $\bm{\alpha}: \Zed \to D$ associated with it (denoted by the same name but written
in boldface), is recursively
defined with respect to a finite sequence of valuations $\sigma_k$ as:
\[
\begin{array}{c|c|c}
  \alpha & 0 \leq i < k & i = k
    \\
    \hline
    x & \bm{x}(i) = \sigma_k(i,x) & \bm{x}(k) = \sigma_k(k,x)\\
    \aX\alpha' & \,\,\bm{\alpha}(i) = \bm{\alpha'}(i+1) & \,\,\bm{\alpha}(k) = \sigma_k(k+|\alpha'|+1,x_{\alpha'})  \\
\end{array}
\]
\[
\begin{array}{c|c|c}
  \alpha & 0 < i \leq k+1 & i = 0
    \\
    \hline
    \aY\alpha' & \,\,\bm{\alpha}(i) = \bm{\alpha'}(i-1)   & \,\,\bm{\alpha}(0) =\sigma_k(|\alpha'|-1,x_{\alpha'})\\
\end{array}
\]
The conjunction of the above subformulae gives formula $|ArithConstraints|_k$.
Implementing $|ArithConstraints|_k$ is straightforward.
In fact, the assignments of values to variables are defined by the
interpretation of the symbols of the QF-EU$\D$ formula.
The values of variables $x$ at positions before $0$ and $k$, i.e. in intervals
$\interval{\lfloor\phi\rfloor}{-1}$ and $\interval{k+1}{k+\lceil\phi\rceil}$, are defined by
means of the values of terms $\alpha=\aX^i x$ and $\alpha=\aY^i x$.
For instance, the value of $x$ at position $0 > i\geq \lfloor\phi\rfloor$ 
is $\sigma_k(i, x)$, but it is defined by the assignment for term
$\alpha=\aY^i x$ at position $0$.

\paragraph{Encoding relations}\label{subsection-encoding-relations}

Formula $|PropConstraints|_k$ encodes atomic subformulae $\theta$
containing relations over a.t.t.'s. 
Let $R$ be an $n$-ary relation of $\mathcal{R}$ that appears in $\phi$,
and $\alpha_1,\dots \alpha_n$ be a.t.t.'s. 
We introduce a \emph{formula predicate}
$\bm{\theta}: \Nat \rightarrow \set{true,false}$ --- that is, a unary uninterpreted predicate
denoted by the same name as the formula but written in boldface --- for all 
$R(\alpha_1,\dots,\alpha_n)$ in $\phi$:
\[
\begin{array}{c|c}
  \theta & 0 \leq i \leq k+1 \\
  \hline
  R(\alpha_1,\dots,\alpha_n) & \,\,\bm{\theta}(i) \Leftrightarrow R(\bm{\alpha_1}(i),\ldots ,\bm{\alpha_n}(i))\\
  \neg R(\alpha_1,\dots,\alpha_n) & \,\,\bm{\theta}(i) \Leftrightarrow \neg R(\bm{\alpha_1}(i),\ldots ,\bm{\alpha_n}(i))
\end{array}
\]

\paragraph{Encoding formulae}

The truth value of a CLTLB formula is defined with respect to the
truth value of its subformulae. 
We associate with each subformula $\theta$ a
formula predicate 
$\bm{\theta}: \Nat \rightarrow \set{true, false}$.
When the subformula $\theta$ holds at instant $i$ then $\bm{\theta}(i)$ holds.
As the length of paths is fixed to $k+1$ and all paths start from $0$,
formula predicates are actually subsets of $\{0, \dots, k+1\}$.
Let $\theta$ be a subformula of $\phi$ and $p$ a propositional letter, formula predicate
$\bm{\theta}$ is recursively defined as:
\[
\begin{array}{c|c}
  \theta & 0 \leq i \leq k+1 \\
  \hline
  p & \bm{p}(i)\\
  \neg p & \,\,\bm{\theta}(i) \Leftrightarrow \neg \bm{p}(i)\\
  \psi_1 \wedge \psi_2 & \,\,\bm{\theta}(i) \Leftrightarrow \bm{\psi_1}(i) \wedge \bm{\psi_2}(i)\\
  \psi_1 \vee \psi_2 &  \,\,\bm{\theta}(i) \Leftrightarrow \bm{\psi_1}(i) \vee \bm{\psi_2}(i)
\end{array}
\]
The conjunction of the formulae above is also part of formula $|PropConstraints|_k$.
The temporal behavior of future and past operators is encoded in formula \linebreak $|TempConstraints|_k$ by using their traditional
fixpoint characterizations.
More precisely, $|TempConstraints|_k$ is the conjunction of the following formulae, for each temporal subformula $\theta$:
\[
\begin{array}{c|c}
  \theta & 0 \leq i \leq k \\
  \hline
  \X\psi & \bm{\theta}(i) \Leftrightarrow \bm{\psi}(i+1) \\
  \psi_1\U\psi_2 & \bm{\theta}(i)\Leftrightarrow(\bm{\psi_2}(i) \vee (\bm{\psi_1}(i) \wedge
\bm{\theta}(i+1)))\\
  \psi_1\R\psi_2 & \bm{\theta}(i)\Leftrightarrow(\bm{\psi_2}(i) \wedge (\bm{\psi_1}(i) \vee
\bm{\theta}(i+1)))\\
\end{array}
\]

\[
\begin{array}{c|c|c}
  \theta& 0 < i \leq k+1 & i = 0\\
  \hline
  \Y\psi & \bm{\Y\psi}(i) \Leftrightarrow \bm{\psi}(i-1) & false\\
  \psi_1\Snc\psi_2 &
  \bm{\theta}(i) \Leftrightarrow (\bm{\psi_2}(i)\vee (\bm{\psi_1}(i) \wedge
  \bm{\theta}(i-1)))&
  \bm{\theta}(0) \Leftrightarrow \bm{\psi_2}(0)  \\
  \psi_1\T \psi_2 &
  \bm{\theta}(i) \Leftrightarrow (\bm{\psi_2}(i) \wedge  (\bm{\psi_1}(i) \vee
  \bm{\theta})(i-1))&
    \bm{\theta}(0) \Leftrightarrow  \bm{\psi_2}(0)
\end{array}
\]
%

\paragraph{Encoding periodicity}

To represent ultimately periodic sequences of symbolic valuations we
use a positive integer variable $\bm{loop} \in \interval{1}{k}$ that captures the position in which the loop starts in $sv_0 sv_1\dots
sv_{loop-1}(sv_{loop}\dots sv_k)^\omega$.
Informally, if the value of variable $\bm{loop}$ is $i$, then there exists a loop which starts at $i$.
To encode the loop we require $sv_{loop-1}=sv_k$; this is achieved through the following formula $|LoopConstraints|_k$, which ranges over all relations $R \in\mathcal{R}$ and all terms in $\terms(\phi)$:
\[
\bigwedge_{\begin{array}{c}
            \theta = R(\alpha_1,\dots,\alpha_n)
            \\
            R \in \mathcal{R}, \alpha_1,\dots,\alpha_n \in \terms(\phi)
           \end{array}}
\bm{\theta}(\bm{loop}-1) = \bm{\theta}(k).
\]

{\em Last state constraints} (captured by formula $|LastStateConstraints|_k$) define the equivalence between
the truth values of the subformulae of $\phi$ at position $k+1$ and those at the position indicated by the $\bm{loop}$ variable,
since the former position is representative of the latter along periodic paths.
These constraints have a similar structure as those in the Boolean encoding of \cite{BHJLS06};
for brevity, we consider only the case for infinite periodic words, as the case for finite words can be easily achieved.
Hence, last state constraints are introduced through the following formula (where $sub(\phi)$ indicates the set of subformulae of $\phi$) by adding only \emph{one} constraint for each
subformula $\theta$ of $\phi$.
\[
  \begin{array}{l}
      \bigwedge_{\theta \in sub(\phi)}\bm{\theta}(k+1) \iFF \bm{\theta}(\bm{loop}). 
  \end{array}
\]

\paragraph{Eventualities for $\U$ and $\R$}

To correctly define the semantics of $\U$ and $\R$, their
\emph{eventualities} have to be accounted for.
Briefly, if $\psi_1\U\psi_2$ holds at $i$, then $\psi_2$ eventually holds
in some $j\geq i$; if $\psi_1\R\psi_2$ does not hold at $i$, then $\psi_2$
eventually does not hold in some $j\geq i$.
\textit{Along finite paths of length $k$, eventualities must hold between $0$
and $k$.
Otherwise, if there is a loop, an eventuality may hold within the loop.}
The Boolean encoding of \cite{BHJLS06} introduces $k$ propositional variables
for each subformula $\theta$ of $\phi$ of the form $\psi_1\U\psi_2$ or $\psi_1\R\psi_2$ (one for each $1 \leq i \leq k$), which represent the eventuality of $\psi_2$
implicit in the formula. 
Instead, in the QF-EU$\mathcal{D}$ encoding, only \emph{one} variable $\bm{j_{\psi_2}} \in
D$
is introduced for each $\psi_2$ occurring in a subformula
$\psi_1\U\psi_2$ or $\psi_1\R\psi_2$. 
\[
\begin{array}{c|c}
    \theta & \\ 
    \hline
    \psi_1\U\psi_2 &
    \begin{array}{l}
    \bm{\theta}(k) \Rightarrow
     \bm{loop} \leq \bm{j_{\psi_2}} \leq k \wedge \bm{\psi_2}(\bm{j_{\psi_2}})
    \end{array}
    \\
    \psi_1\R\psi_2 &
    \begin{array}{l}
    \neg\bm{\theta}(k) \Rightarrow
     \bm{loop} \leq \bm{j_{\psi_2}} \leq k \wedge \neg \bm{\psi_2}(\bm{j_{\psi_2}})
    \end{array}
  \end{array}
\]
The conjunction of the constraints above for all subformulae $\theta$
of $\phi$ constitutes the formula $|Eventually|_k$.

The complete encoding $|\phi|_k$ of $\phi$ consists of the logical
conjunction of all above components, together with $\phi$ evaluated at the
first instant of time. 

\subsection{Correctness of the BSP encoding}
\label{section-correctness}

%
To prove the correctness of the encoding defined in Section~\ref{section-encoding},
we first introduce two properties, which reduce CLTLB($\D,\mathcal{R}_0$) to CLTLB($\D$) without $\aY$ operators.
This allows us to base our proof on the automata-based construction for CLTLB($\D$) of \cite{DD07}.
In particular, the two reductions are essential to take advantage of Proposition
\ref{prop-completion} and Lemma \ref{lemma-condC} of Section \ref{sec:cltlb}, to
define a decision procedure for the bounded satisfiability problem of
Section \ref{section-bsp}.
The properties are almost obvious, hence we only provide the intuition behind their proof (see \cite{BFMPRS12-arxiv} for full details).

%



\begin{subproperty}\label{theorem-removeAP}
CLTLB($\D,\mathcal{R}_0$) formulae can be equivalently rewritten into CLTLB($\D$) formulae.
\end{subproperty}

According to the definition given in Section \ref{subsection-syntax},
CLTLB($\D$) is the language CLTLB where atomic formulae belong to the
language of constraints in $\D$, which may contain also $0$-ary relations. In this case, atomic
formulae are propositions $p \in \mathcal{R}_0$ or relations 
$R(\alpha_1,\dots, \alpha_n)$.
Any positive occurrence of an atomic proposition $p \in \mathcal{R}_0$ in a
CLTLB formula can be replaced by an equality relation of the form
$x_p=1$.
Then, a formula of CLTLB$ (\D,\mathcal{R}_0)$ can be easily rewritten into a
formula of CLTLB$ (\D)$ preserving the equivalence between them
(modulo the rewriting of propositions in $\mathcal{R}_0$).
We define a rewriting function $\noprop$ over formulae such that
$\pair{\pi'}{\sigma'},0 \models \phi$ if, and only if, $\pair{\pi}{\sigma},0
\models \noprop(\phi)\wedge \psi$ where $\sigma$ is the same as $\sigma'$ except for
new fresh variables $x_p$ representing atomic propositions, and $\psi$ is
a formula restricting the values of variables $x_p$ to $\set{0,1}$.

For instance, let $\phi$ be the formula $\G(p \Rightarrow \F(\aX x < y \wedge q))$, where the ``eventually'' ($\F$) and ``globally'' ($\G$) operators are defined as usual.
The formula obtained by means of rewriting $\noprop$ is 
$$\G(x_p=1 \Rightarrow \F(\aX x < y \wedge x_q=1)) \wedge 
\G\left(\begin{array}{c}
(x_p=1\vee x_p=0) \\ \wedge \\
(x_q=1\vee x_q=0)
\end{array}\right).
$$

Note that formula $\noprop(\phi)$ does not contain any propositional letters, so in a model $\pair{\pi}{\sigma}$ component $\pi$ associates with each instant the empty set.
From now on we will consider only CLTLB$(\D)$ formulae without propositional letters; hence, given a propositional letter-free formula $\phi$,
we will write $\sigma, 0 \models \phi$ instead of $\pair{\pi}{\sigma},0 \models \phi$.




\begin{subproperty}\label{theorem-removeY}
CLTLB($\D$) formulae can be equivalently rewritten into CLTLB($\D$)
  formulae without $\aY$ operators.
\end{subproperty}

Let $\shiftleft:\text{CLTLB($\D$)} \rightarrow \text{CLTLB($\D$)}$ be the following mapping, which transforms (by ''shifting to the left'') every
formula $\phi$ into an equisatisfiable formula that does not contain any occurrence
of the $\aY$ operator. Formula $\shiftleft(\phi)$ is identical to $\phi$ except that all a.t.t.'s
of the form $\aX^ix$ in $\phi$ are replaced 
by $\aX^{i-\lfloor\phi \rfloor}x$, while all a.t.t.'s of the
form $\aY^ix$ are replaced by $\aX^{-i-\lfloor\phi \rfloor}x$. 
The latter replacement avoids 
negative indexes (since if $\phi$ contains a.t.t.'s of the form $\aY^i x$, 
then 
$\lfloor \phi \rfloor < 0$).
The $\shiftleft$ function can be naturally extended to symbolic valuations (i.e, sets of atomic constraints) and sequences $\rho$ thereof.

%
%
%
As a consequence, given a CLTLB($\D$) formula $\phi$, it is easy to see that $\aY$ does not occur in $\shiftleft(\phi)$.
The equisatisfiability of formulae $\phi$ and $\shiftleft(\phi)$ is guaranteed by moving the origin of $\phi$ by $-\lfloor\phi \rfloor$ instants in the past.
Since only $\aX$ occurs in $\shiftleft(\phi)$, then models for CLTLB($\D$) formulae without $\aY$ are now sequences of $\D$-valuations $\sigma: \N \times V \rightarrow D$.

\begin{proposition}\label{prop-removeY}
Let $\phi$ be a CLTLB($\D$) formula, then $\sigma,0 \models \phi$ iff $\sigma, \lfloor\phi \rfloor \models \shiftleft(\phi)$.
\end{proposition}
\begin{corollary}\label{corollary-shift-models}
Let $\rho \in SV(\phi)^\omega$ be a sequence of symbolic valuations.
Then,
\[
\begin{array}{c}
\sigma, 0 \models \rho\quad \textrm{ iff }\quad \sigma, \lfloor\phi \rfloor \models \shiftleft(\rho) \\
\rho, 0 \symodels \phi\quad \textrm{ iff }\quad \shiftleft(\rho), 0 \symodels \shiftleft(\phi).
\end{array}
\]
\end{corollary}



We now have all necessary elements to prove the correctness of our encoding.
We first provide the following three equivalences, which are proved by showing the implications depicted in Figure \ref{fig-sat-equivalences}, where $\A_s\times\A_{\ell}$ is the automaton recognizing symbolic models of $\shiftleft(\phi)$:
\begin{enumerate}
\item Satisfiability of $|\phi|_k$ is equivalent to the existence of ultimately periodic runs of automaton $\A_s\times\A_{\ell}$.
\item $k$-satisfiability is equivalent to the existence of ultimately periodic runs of automaton $\A_s\times\A_{\ell}$.
\item $k$-satisfiability is equivalent to the satisfiability of $|\phi|_k$.
\end{enumerate}
Then we draw, by Proposition \ref{proposition-BSP-to-SAT-I}, the connection between $k$-satisfiability and satisfiability for formulae over  constraint systems satisfying the completion property.
In Section \ref{subsection-A_C}, thanks to Proposition \ref{proposition-BSP-to-SAT-II}, we extend the result to constraint system IPC$^*$, which does not have the completion property. 
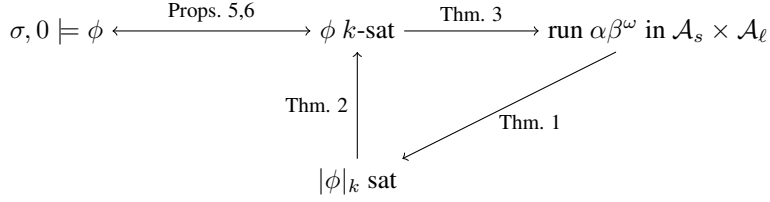
\begin{figure}[!ht]
\centering
\begin{tikzpicture}[auto]
   \node at (0,0) (A) {$|\phi|_k$ sat};
   \node at (0,2) (B) {$\phi$ $k$-sat};
   \node at (4,2) (C) {run $\alpha\beta^\omega$ in $\A_s\times\A_\ell$};
   \node at (-4,2) (D) {$\sigma,0\models \phi$};
   \path[->] (A) edge node {\footnotesize Thm. \ref{theorem-eq-enc-BSP}} (B);
   \path[<-] (A) edge node[below] {\footnotesize\quad \quad Thm. \ref{theorem-eq-encod-run}} (C);
   \path[->] (B) edge node {\footnotesize Thm. \ref{theorem-eq-BPS-run}} (C);
   \path[<->] (D) edge node {\footnotesize Props. \ref{proposition-BSP-to-SAT-I},\ref{proposition-BSP-to-SAT-II}} (B);
\end{tikzpicture}
\caption{Proof schema.}
\label{fig-sat-equivalences}
\end{figure}

Before tackling the theorems of Figure \ref{fig-sat-equivalences}, we provide the definition of models for QF-EU$\D$ formulae $|\phi|_k$ built according to the encoding of Section \ref{section-encoding}.
More precisely, a model $\mathcal{M}$ of $|\phi|_k$ is a pair $\pair{D}{{\mathcal{I}}}$ where $D$ is the domain of interpretation of $\D$, and $\mathcal{I}$ maps
\begin{itemize}
\item each function symbol $\bm{\alpha}$ onto a function associating, for
  each position of time, an element in domain $D$,
  $\mathcal{I}(\bm{\alpha}):\Nat\rightarrow D$;
\item each predicate symbol $\bm{\theta}$ onto a function associating, for
  each position of time, an element in $\set{true,false}$,
  $\mathcal{I}(\bm{\theta}):\Nat\rightarrow \set{true, false}$.
\end{itemize}
Note that mapping $\mathcal{I}$ trivially induces a finite sequence of $\D$-valuations $\sigma_k:\set{\lfloor\phi\rfloor,\dots,k+\lceil\phi\rceil}:V \rightarrow D$.


We start by showing that the existence of ultimately periodic runs of automaton $\A_s\times\A_{\ell}$ implies the satisfiability of $|\phi|_k$.

\begin{theorem}\label{theorem-eq-encod-run}
Let $\phi\in$ CLTLB($\D$) with $\N$ definable in $\D$ together with
the successor relation.
If there exists an ultimately periodic run $\rho=\alpha\beta^\omega$ ($|\alpha\beta|=k+1$) of
$\A_s\times\A_{\ell}$ accepting symbolic models of $\shiftleft(\phi)$, then $|\phi|_k$ is satisfiable with
respect to $k \in \N$ .
\end{theorem}

In the following proof, we use 
the generalized B\"uchi automaton obtained by the standard
construction of \cite{vw}, in the version of \cite{DD07}.
Let $\phi'$ be a CLTLB($\D$) formula (without the $\aY$ modality over
terms).
The closure of $\phi'$, denoted  $cl(\phi')$, is the smallest negation-closed set containing all subformulae of $\phi$.
An \emph{atom} $\Gamma \subseteq cl(\phi')$ is a subset of formulae of
$cl(\phi')$ that is maximally consistent, i.e., such that, for each
subformula $\xi$ of $\phi'$, either $\xi \in \Gamma$ or $\neg \xi \in
\Gamma$.
A pair $(\Gamma_1,\Gamma_2)$ of atoms is \emph{one-step temporally consistent}
when: 
\begin{itemize}
\item for every $\X\psi \in cl(\phi')$, then $\X\psi \in \Gamma_1 \siff \psi \in \Gamma_2$,

\item for every $\Y\psi \in cl(\phi')$, then $\Y\psi \in \Gamma_2 \siff \psi \in \Gamma_1$,

\item if $\psi_1\U \psi_2 \in \Gamma_1$, then $\psi_2 \in \Gamma_1 \text{ or both } \psi_1 \in \Gamma_1 \text{ and } \psi_1\U \psi_2 \in \Gamma_2$,

\item if $\psi_1\Snc \psi_2 \in \Gamma_2$, then $\psi_2 \in \Gamma_2 \text{ or both } \psi_1 \in \Gamma_2 \text{ and } \psi_1\Snc \psi_2 \in \Gamma_1$.
\end{itemize}

The automaton $\A_s = (SV(\phi'), Q, Q_0, \eta, F)$ is then
defined as follows:
\begin{itemize}
\item $Q$ is the set of atoms;

\item $Q_0 = \{\Gamma \in Q : \phi' \in \Gamma, \Y\psi \notin \Gamma \text{ for all } \psi \in cl(\phi'),  \psi_1\Snc \psi_2 \in \Gamma \text{ iff } \psi_2 \in \Gamma\}$;

\item $\Gamma_1 \xrightarrow{sv} \Gamma_2 \in \eta$ iff
  \begin{itemize}
  \item $sv \symodels \Gamma_1$
  \item $(\Gamma_1,\Gamma_2)$ is one-step consistent;
  \end{itemize}
\item $F=\{F_1, \dots ,   F_p\}$, where $F_i = \{\Gamma \in Q \mid
  \psi_i \U \zeta_i \notin \Gamma \; \text{ or } \; \zeta_i \in \Gamma\}$ and $\{\psi_1 \U \zeta_1, \dots ,\psi_p \U \zeta_p\}$ is the set of Until formulae occurring in $cl(\phi')$.
\end{itemize}

\begin{proof}
We prove that if there is a run in
$\A_s\times\A_\ell$ accepting $\shiftleft(\phi)$, then formula $|\phi|_k$ is
satisfiable (we assume the rewriting obtained through $\noprop$).
Suppose there exists an ultimately periodic
symbolic model of length $k+1$
which is accepted by $\A_s\times \A_\ell$.
It is a locally consistent sequence of symbolic valuations, $\rho =
\alpha\beta^\omega$ of the form:
\[
\rho = sv_0\dots sv_{loop-1}(sv_{loop} \dots sv_k)^\omega
\]
such that $\rho \in \Lng(\A_s\times\A_\ell)$ (for simplicity, and without loss of generality, we assume that $sv_{loop-1} = sv_k$).
$\rho$ is recognized by a periodic run of $\A_s\times \A_\ell$ of the
form\footnote{For reasons of clarity, we avoid some details of product automaton $\A_s\times \A_\ell$, which are however inessential in the proof.}:
\[
\upsilon=\pair{\Gamma_0}{sv_0}\dots
\pair{\Gamma_{loop-1}}{sv_{loop-1}}(\pair{\Gamma_{loop}}{sv_{loop}}\dots
\pair{\Gamma_{k}}{sv_{k}})^\omega.
\]
For each subformula $\psi_i \U \zeta_i$ occurring in $\phi$,
subrun $\pair{\Gamma_{loop-1}}{sv_{loop-1}}\pair{\Gamma_{loop}}{sv_{loop}}\dots
\pair{\Gamma_{k}}{sv_{k}}$ visits control states of the set $F_i$, thus
witnessing the acceptance condition of $\A_s$.
%
From $\upsilon$ we build
run $\gamma$ of $\A_s$:
\[
\gamma=\Gamma_0\dots \Gamma_{loop-1}(\Gamma_{loop}\dots \Gamma_k)^\omega.
\]
In particular, $\rho$ is defined by the projection on the alphabet of
$SV(\shiftleft(\phi))$ of the subformulae occurring in every $\Gamma_i$, for $0 \leq i
\leq k$.
Sequence $\rho$ and its accepting run $\gamma$ can be translated by
means of $\shiftleft^{-1}$ to obtain a symbolic model for $\phi$.
In particular, because $\rho,0 \symodels \shiftleft(\phi)$ then we obtain, by Corollary \ref{corollary-shift-models}, $\shiftleft^{-1}(\rho),0
\symodels \phi$.
Similarly, by shifting all formulae in atoms of $\gamma$, we obtain an
accepting run $\shiftleft^{-1}(\gamma)$ for $\phi$.
The model for $|\phi|_k$ is given by the
truth value of all the subformulae
in each $\shiftleft^{-1}(\Gamma_i)$ and the values of variables occurring in $\phi$ can be defined as explained later.
In particular, we need to complete interpretation $\mathcal{I}$ for
uninterpreted predicate and functions formulae: given a
position $0\leq i\leq k$, for all subformulae $\theta \in cl(\phi)$ we
define
\begin{itemize}
\item $\mathcal{I}(\bm{\theta})(i) = true$ iff $\theta \in \shiftleft^{-1}(\Gamma_i)$,
\item $\mathcal{I}(\bm{\theta})(i) = false$ iff $\neg \theta \in \shiftleft^{-1}(\Gamma_i)$.
\end{itemize}
The truth value of subformulae $\psi \R \zeta$ and $\psi \T \zeta$ is derived by duality.
To complete the interpretation of subformulae at position $k+1$ we can
use values from position $loop$: $\mathcal{I}(\bm{\theta})(k+1) =
\mathcal{I}(\bm{\theta})(loop)$.
Note that by taking truth values of subformulae $\theta \in
cl(\phi)$ from atoms $\shiftleft^{-1}(\Gamma_i)$, $|propConstraints|_k$ are trivially
satisfied (atoms are defined by using the same Boolean closure in $|propConstraints|_k$).
The sequence $\rho$ of symbolic valuations is consistent and all the a.t.t.'s
in the encoding of $|\phi|_k$ can be uniquely defined by considering
at each position $i$ a symbolic valuation $\shiftleft^{-1}(sv_i)$.
Consider the sequence $\rho' = sv_0\dots sv_{loop-1}(sv_{loop}
\dots sv_k)sv_{loop}$.
Following \cite[Lemma 5.2]{DD07}, we can build an edge-respecting assignment of values in $D$ for the finite graph $G_{\shiftleft^{-1}(\rho')}$, which associates, for each for each variable $x\in V$ and for each position
$\lfloor \phi \rfloor \leq i\leq k+1+\lceil \phi \rceil$, a value $\sigma_k(i,x)$.
We exploit assignment $\sigma_k(i,x)$ to define $\mathcal{I}(\bm{\alpha})$, with $\alpha \in \terms(\phi)$, in the following way (where $x_\alpha$ is the variable in $\alpha$):
\[
\mathcal{I}(\bm{\alpha})(i)=\sigma_k(i+|\alpha|, x_\alpha)
\]
for all $0\leq i\leq k+1$.
Then, formulae $|ArithConstraints|_k$ are satisfied.
Since run $\upsilon$ is ultimately periodic, then control state
$\pair{\Gamma_{loop}}{sv_{loop}}$ is visited at position $k+1$.
It witnesses the satisfaction of
$|LastStateConstraints|_k$ formulae, which prescribe that $\theta_{k+1} \iFF \theta_{loop}$
for all $\theta \in cl(\phi)$.
Finally, let us consider $|Eventually|_k$ formulae.
If subformula $\varphi = \psi\U\zeta$ belongs to atom $\Gamma_k$,
then there exists a position $j \geq k$ such that $\zeta$ holds in $j$.
Since the model is periodic, then $k \leq j
\leq k + |\beta|$, i.e., $\bm{j}_\zeta = j - |\beta|$ is a position such that $loop \leq \bm{j}_\zeta \leq k$.
Moreover, if $\neg (\psi\R\zeta) = \neg\psi \U \neg\zeta$ belongs to
$\Gamma_k$ then there exists a position $j\geq k$ such that $\neg \zeta$ holds in $j$.
As in the previous case $loop \leq \bm{j}_\zeta \leq k$.
Hence, the $|Eventually|_k$ formulae are satisfied.
The initial atom $\Gamma_0$ is such that $\Y\varphi \not\in \Gamma_0$ and if
$\psi \Snc\zeta \in \Gamma_0$ then $\zeta \in \Gamma_0$, which witnesses the
encoding of subformulae $\Y\psi$ and $\psi\Snc\zeta$
at 0, i.e., $\theta_0 \iFF \perp$ and $\theta_0
\iFF \zeta_0$, respectively.
\end{proof}

We now prove the second implication, which draws the connection between the encoding and the $k$-satisfiability problem.

\begin{theorem}\label{theorem-eq-enc-BSP}
Let $\phi\in$ CLTLB($\D$) with $\N$ definable in $\D$ together with
the successor relation. If $|\phi|_k$ is satisfiable, then formula $\phi$ is $k$-satisfiable with respect
to $k \in \N$.
\end{theorem}
\begin{proof}
We prove the theorem by showing that formula $|\phi|_k$ defines ultimately periodic symbolic models $\rho=\alpha\beta^\omega$ for formula $\phi$ such that $\sigma_k, 0 \models_k \alpha\beta$ and $\rho,0 \symodels \phi$.
Note that the encoding of $|\phi|_k$ defines
precisely the truth value of all subformulae $\theta$ of $\phi$ in instants $i \in \interval{0}{k}$.
Then, if $|\phi|_k$ is satisfiable, given an $i \in \interval{0}{k}$, the set of all subformulae
\[
\Gamma_i=\{\varphi \in cl(\phi) \mid \text{ if } \bm{\theta}(i) \text{ holds then } \varphi = \theta \text{, else }
\varphi=\neg\theta\}
\]
is a maximal consistent set of subformulae of $\phi$.
We have $\bm{loop}\in \interval{1}{k}$.
The sequence of sets $\Gamma_i$ for $0\leq i\leq k$ is an ultimately
periodic sequence of 
maximal consistent sets
due to formulae $|LastStateConstraints|_k$ and $|LoopConstraints|_k$.
We write $\Gamma|_{A}$ to denote the projection of $\D$-constraints in $\Gamma$
on symbols of the set $A$; e.g., if
$A=\set{R_1,R_2}$ then $\set{R_1(x,y),R_2(\aX x, \aY
  x),\theta_1,\theta_2}|_A=\set{R_1(x,y),R_2(\aX x, \aY x)}$.
The sequence of atoms is
\[
\gamma=\Gamma_0 \dots \Gamma_{loop-1}
\left(\Gamma_{loop}  \dots,
\Gamma_k\right)^\omega
\]
and such that $\Gamma_{loop-1}|_{\mathcal{R}}$ is equal to the set of
relations of $\Gamma_{k}|_{\mathcal{R}}$ by $|LoopConstraints|_k$ formulae.
Moreover, by $|LastStateConstraints|_k$ we have
$\Gamma_{k+1}=\Gamma_{loop}$.

By Lemma \ref{lemma-B}, from the bounded sequence $\sigma_k$ of $\D$-valuations induced by $\mathcal{I}$, we
have a unique locally
consistent finite sequence of symbolic valuations $\alpha\beta$ such that
$\sigma_k,0\models_k\alpha\beta$.
Formula $|LoopConstraints|_k$ witnesses ultimately periodic sequences
of symbolic valuations $\rho$ because it is defined over the set of relations
in $\mathcal{R}$ and all terms of the set $\terms(\phi)$:
\[
\rho = \alpha\beta^\omega = sv_0\dots sv_{loop-1}(sv_{loop} \dots sv_k)^\omega
\]
such that $sv_{loop-1}=sv_k$.

We call $\rho^i$ the suffix of $\rho$ that starts from position $i \geq 0$.
By structural induction on $\phi$ one can prove that for all $0\leq i
\leq k+1$, for all subformulae $\theta$ of $\phi$, $\bm{\theta}(i)$ holds (i.e., $\theta \in \Gamma_i$) if, and only if,
\begin{itemize}
\item
$ \rho^i, 0 \symodels \theta$ for $\theta$ of the form $R, \X, \U, \R$;
\item $(sv_{0} \dots sv_i), i \symodels \theta$ for $\theta$ of the form $\Y, \Snc, \T$.
\end{itemize}
Then, since by hypothesis $\bm{\phi}(0)$ holds, we have that $\rho,0 \symodels \phi$.


The base case is the unique fundamental part of the proof because the inductive step over temporal modalities is rather standard.
Let us consider a relation formula $\theta$ of the form $R(\alpha_1,\dots,\alpha_n)$ where, for all $1 \leq j \leq n$, $\alpha_j \in \terms(\phi)$.
We have to show that $\bm{\theta}(i)$ holds if, and only if, $sv_i \symodels \theta$.
We have that $\bm{\theta}(i)$ holds if, and only if, $\sigma_k,i \models_k \theta$; since, by Lemma \ref{lemma-B}, $\sigma_k,i \models \theta$ if, and only if, the symbolic valuation $sv_i$ induced by $\sigma_k$ at $i$ includes $\theta$, we have by definition $sv_i \symodels \theta$.

%

We omit the inductive step, which is standard and is reported in \cite{BHJLS06} and \cite{PMS12}, since we use the same operators with the same encodings.
\end{proof}

Finally, the next theorem draws a link between $k$-satisfiability and the
existence of an ultimately periodic run in automaton $\A_s\times\A_\ell$.

\begin{theorem}\label{theorem-eq-BPS-run}
Let $\phi\in$ CLTLB($\D$) with $\N$ definable in $\D$ together with
the successor relation. If formula $\phi$ is $k$-satisfiable with respect
to $k \in \N$, then there exists an ultimately
periodic run $\rho=\alpha\beta^\omega$ of $\A_s\times\A_{\ell}$, with $|\alpha\beta| = k+1$, accepting symbolic
models of $\shiftleft(\phi)$.
\end{theorem}
\begin{proof}
By definition, if $\phi$ is $k$-satisfiable so is $\shiftleft(\phi)$, and there is an ultimately
periodic symbolic model $\rho=\alpha\beta^\omega$ such that
$\rho,0\models \shiftleft(\phi)$.
By Lemma~\ref{lemma-B}, $\rho$ is locally consistent because there
exists a $k$-bounded model $\sigma_k$ such that $\sigma_k \models_k
\alpha\beta$.
Therefore, $\rho\in\Lng(\A_s\times \A_\ell)$.
\end{proof}

As explained in Section \ref{section-symbolic}, each automaton involved in the definition of
$\A_\phi$ has the function of \lq\lq filtering\rq\rq\ sequences of symbolic valuations
so that 1) they are locally consistent, 2) they satisfy an LTL
property and 3) they admit a (arithmetic) model.
As mentioned in Section \ref{sec:cltlb}, for constraint systems that have the completion property local consistency is a sufficient and necessary condition for admitting a model.
For these constraint systems $\A_\phi$ is exactly
automaton $\A_s \times \A_\ell$, and from
Proposition \ref{prop-completion} and Theorem \ref{theorem-eq-enc-BSP} we obtain the following result.

\begin{proposition}\label{proposition-BSP-to-SAT-I}
Let $\phi\in$ CLTLB($\D$) with $\N$ definable in $\D$ together with
the successor relation and satisfying the completion property. Formula $\phi$ is $k$-satisfiable with respect
to some $k \in \N$ if, and only if, 
there exists a model $\sigma$ such that $\sigma,0 \models \phi$.
\end{proposition}
\begin{proof}
Suppose formula $\phi$ is $k$-satisfiable.
Then, by Theorem \ref{theorem-eq-BPS-run}, there is a symbolic model $\rho = \alpha\beta^\omega$ such that $\rho, 0 \symodels \shiftleft(\phi)$.
By Proposition \ref{prop-completion} $\rho$ admits a model $\sigma'$, i.e., such that $\sigma', 0 \models \shiftleft(\phi)$.
By Corollary \ref{corollary-shift-models}, we have $\sigma', -\lfloor\phi \rfloor \models \phi$, so the desired $\sigma$ is simply $\sigma'$ translated of $\lfloor\phi \rfloor$.


Conversely, if formula $\phi$ is satisfiable, then automaton $\A_{\shiftleft(\phi)}$
recognizes a nonempty language in $SV(\shiftleft(\phi))^\omega$.
Hence, there is an ultimately periodic, locally consistent, sequence of symbolic valuations $\rho = \alpha\beta^\omega$, with $|\alpha\beta| = k+1$, which is accepted by automaton $\A_{\shiftleft(\phi)}$.
Then, the model $\sigma_k$ that shows the $k$-satisfiability of $\phi$ is built considering prefix $\alpha\beta$, by defining an edge-respecting labeling of graph $G_{\alpha\beta}$.
\end{proof}

When constraint systems do not have the completion
property,  locally consistent symbolic models $\rho$ recognized by automaton $\A_s\times
\A_\ell$ may not admit arithmetical models $\sigma$ such that $\sigma\models \rho$.
However, as mentioned in Section \ref{subsec:completion}, for some constraint systems $\D$, it is possible to define a condition $C$ over symbolic models
such that $\rho \in \Lng(\A_s\times\A_\ell)$ satisfies $C$ if, and only if
$\rho$ admits a model.
We tackle this issue in the next section.

\section{Bounded Satisfiability of CLTLB$(\text{IPC}^*)$}
\label{subsection-A_C}

When $\D$ is IPC$^*$, Proposition \ref{proposition-BSP-to-SAT-I} does not apply since, by Lemma \ref{lem:completion}, IPC$^*$ does not have the completion property.
However, as shown by Lemma \ref{lemma-condC}, ultimately periodic symbolic models of CLTLB(IPC$^*$) formulae admit arithmetic model if, and only if, they obey the condition captured by Property \ref{property-C}.
In this section, we define a simplified condition of (non) existence of arithmetical models for ultimately periodic symbolic models of CLTLB(IPC$^*$) formulae, and we show its equivalence with Property \ref{property-C}.
Then, we provide a bounded encoding through QF-EU$\D$ formulae (where $\D$ embeds $\Nat$ and the successor function) for the new condition, and we define a specialized version of Proposition \ref{proposition-BSP-to-SAT-I} for $\D = \text{IPC}^*$.
Finally, we introduce simplifications to the encoding that can be applied in special cases.

Let $\rho$ be a symbolic model for CLTLB(IPC$^*$) formula $\phi$.
To devise the simplified condition equivalent to Property \ref{property-C}, we provide a specialized version of graph $G_\rho$ where points are identified by their relative position within symbolic valuations.
We introduce the notion of point $p=\triple{x}{j}{h}$ in $\rho$ which we use to
identify a variable or a constant $x \in V \cup \const(\phi)$ at position $h$ within  symbolic valuation
$\rho(j)$; i.e., we refer to variable
$x$, or constant $c$, at position $j+h$ of the symbolic model $\rho$.
Given a point $p=\triple{x}{j}{h}$ of $\rho$, we denote with $\var(p)$
the variable $x$, with $\sv(p)$ the symbolic valuation $j$ (with $\sv(p) \geq 0$), and with
$\shift(p)$ the position $h$ of $x$ within the $j$-th symbolic valuation (with $\mathit{shift}(p) \in \interval{\lfloor\phi\rfloor}{\lceil\phi\rceil}$); 
also, $x(j+h)$ is the value of variable $x$ in position $h$ of the $j$-th symbolic valuation of $\rho$.
Given a symbolic model $\rho$, we indicate by $P_\rho$ the set of points of $\rho$.

Different triples can refer to equivalent points.
For example, variable $x$ in position $2$ of symbolic valuation $4$ (i.e., $\triple{x}{4}{2}$) is the same as $x$ in position $1$ of adjacent symbolic valuation $5$ (i.e., $\triple{x}{5}{1}$), and also of $x$ in position $0$ of symbolic valuation $6$ (i.e., $\triple{x}{6}{0}$).
Figures \ref{figure-example-concistency-fwd} and \ref{figure-example-concistency-fwd-II} show examples of equivalent points.
Hence, we need to define an equivalence relation between triples, called \emph{local equivalence}.

\begin{definition} 
For all points $p_1= \triple{x}{j}{h}$, $p_2= \triple{x}{i}{m}$ in $P_\rho$, we say that $p_1$ is \emph{locally equivalent} to $p_2$ if
$j+h = i+m$, with $i,j \geq 0$ and $h,m \in \interval{\lfloor\phi\rfloor}{\lceil\phi\rceil}$.
\end{definition}

\begin{definition}\label{def:lf_lv_peq}
We define the relation $\locfwd \subseteq P_\rho \times P_\rho$. 
Given $p_1=\triple{x}{j}{h}$ and $p_2=\triple{y}{i}{m}$ of $P_\rho$, it is $p_1 \locfwd p_2$ if:
\begin{enumerate} 
\item
$i+m - (j+h) < -\lfloor\phi\rfloor + \lceil\phi\rceil +1$.

\item
$j+h \leq i+m$

\item
$x(j+h) \leq y(i+m)$ 


\end{enumerate}
 Similarly, relations $\locfwds , \locbwd , \locbwds, \approx\; \subseteq P_\rho \times P_\rho$ are defined as above by replacing $\leq$ with, respectively, $<, \geq, >, =$ in Condition 3.
\end{definition}

By Condition 1 of Definition~\ref{def:lf_lv_peq}, for each relation $\sim \in \set{\preccurlyeq, \prec, \approx, \succ, \succcurlyeq}$, 
$p_1\sim p_2$ may hold 
only if  the distance between $p_1$ and $p_2$ is smaller than the size 
$-\lfloor\phi\rfloor + \lceil\phi\rceil + 1$ of a symbolic valuation, i.e., $p_1$ and $p_2$ are ``local'', 
in the sense that they belong either to the same symbolic valuation (i.e., $j = i$) 
or to the common part of ``partially overlapping'' symbolic valuations (see Figures~\ref{figure-example-concistency-fwd} and \ref{figure-example-concistency-fwd-II} 
for examples of partially overlapping symbolic valuations). 
By Condition 2, each relation $\sim$ is a positional precedence, i.e., 
if $p_1 \sim p_2$ then $p_2$ cannot positionally precede $p_1$. 
Condition 3 is well defined on symbolic valuations, since it corresponds to having, in graph $G_\rho$, an arc between $p_1$ and $p_2$ that is labeled with $\sim$. 
The reflexive relations $\locfwd, \locbwd$ have an antisymmetric property, in the sense that 
if $p_1 \locfwd p_2$ and  $p_2 \locfwd p_1$, then $p_1 \approx p_2$ and $p_2 \approx p_1$ (analogously for $\locbwd$): if  $p_1=\triple{x}{j}{h}$ and $p_2=\triple{y}{i}{m}$,  
then $p_1$ and $p_2$ are at the same position $j + h = i + m$ and have the same value $x(j+h)=y(i+m)$.

Notice that the relations $\sim$ are not transitive, because of Condition 1: 
each relation $\sim$ is only ``locally'' transitive, in the sense that if $p_1 \sim p_2$ and $p_2 \sim p_3$, then 
$p_1\sim p_3$ if, and only if, Condition 1 holds for $p_1$ and $p_3$ (i.e., when also $p_1, p_3$ are ``local'', which in general may not be the case).


%
%
%

\begin{definition}
\label{def:locfwdrel}
We say that there is a \emph{local forward} (resp. \emph{local backward}) path from point $p_1$ to point $p_2$ if $p_1 \locfwd p_2$ (resp., $p_1 \locbwd p_2$); the path is called 
\emph{strict} if $p_1 \locfwds p_2$ (resp., $p_1 \locbwds p_2$).
\end{definition}



Obviously, given two points $p_1=\triple{x}{j}{h}$ and $p_2=\triple{y}{i}{m}$ of $P_\rho$ such that $|i+m - (j+h)| < -\lfloor\phi\rfloor + \lceil\phi\rceil +1$, it must be at least one 
of $p_1 \locfwd p_2$, $p_2 \locfwd p_1$, $p_1 \locbwd p_2$, $p_2 \locbwd p_1$; 
if it is both $p_1 \locfwd p_2$ and $p_1 \locbwd p_2$, then $p_1 \approx p_2$, hence $x(j+h) = y(i+m)$.

It is immediate to notice that the local equivalence is a congruence for all relations, e.g., 
if $p_1$ is locally equivalent to $p'_1$ and $p_2$ is locally equivalent to $p'_2$ then $p_1 \locfwd p_2 \iFF p_1' \locfwd p_2'$.
Figures \ref{figure-example-concistency-fwd} and \ref{figure-example-concistency-fwd-II} depict examples of this fact.

We now extend the relations of Definition~\ref{def:locfwdrel} to cope with non-overlapping symbolic valuations. 

\begin{definition}
\label{def:fwdrel}
Relation $\pfwd{\sim}\subseteq  P_\rho \times P_\rho$, for every $\sim\in\set{\preccurlyeq, \approx, \succcurlyeq}$, denotes the transitive closure of $\sim$. 
Relations $\fwds, \bwds\subseteq  P_\rho \times P_\rho$, are defined as follows, for all $p_1,p_2 \in P_\rho$: 

$p_1 \fwds p_2$  if there exist $p',p''\in P_\rho$ such that $p_1\fwd p' \locfwds p'' \fwd p_2$;

$p_1 \bwds p_2$  if there exist $p',p''\in P_\rho$ such that $p_1\bwd p' \locbwds p'' \bwd p_2$.
\end{definition}

%
%
\begin{figure}[tb]
\centering
\begin{tikzpicture}[
  dot/.style={circle,fill=black,minimum size=3pt,inner sep=0pt,
            outer sep=-1pt},
  sv/.style={rectangle,
                                    dashed,
                                    minimum size=1cm},
  ]
  \node[dot] at (1,1) (A) {};
  \node[dot] at (1,2) (B) {};
  \node[dot] at (2,1) (C) {};
  \node[dot] at (2,2) (D) {};
  \node[dot] at (3,1) (E) {};
  \node[dot] at (3,2) (F) {};

  \node[dot, color=gray] at (4,1) (E) {};
  \node[dot, color=gray] at (4,2) (F) {};

  \node[dot, color=gray] at (0,1) (G) {};
  \node[dot, color=gray] at (0,2) (H) {};
  \node[dot, color=gray] at (-1,1) (I) {};
  \node[dot, color=gray] at (-1,2) (L) {};

  \node at (4.3,2) (M) {\footnotesize $p_2$};
  \node at (1,1.3) (M) {\footnotesize $p_1, p_1'$};
  \draw (0.6,0.8) rectangle  (3.2,2.2);

  \draw[dotted] (-1.4,0.6) rectangle  (1.2,2.4);

  \node at (0,0.4) (tag2) {$i-2$};
  \node at (2,0.4) (tag3) {$i$};

  \node at (-2,1) {$y$};
  \node at (-2,2) {$x$};

  \path[-, > = stealth, out=50,in=220, dashed] (A) edge node[above]{\small$\sim$} (F);
\end{tikzpicture}
\caption{Adjacent and overlapping symbolic valuations $\rho(i)$ (solid line) and $\rho(i-2)$ (dotted line) of length 3 (with $-\lfloor\phi\rfloor = \lceil\phi\rceil = 1$), with $p_1=(y,i,-1)$ and $p_1' = (y,i-2,1)$ being locally equivalent. Both $p_1 \pfwd{\sim} p_2$ and $p_1' \pfwd{\sim} p_2$ hold.}
\label{figure-example-concistency-fwd}
\end{figure}
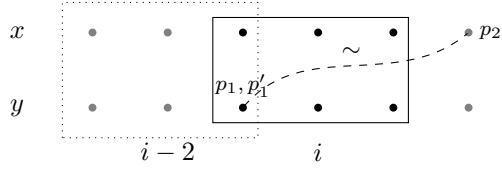
\begin{figure}[ht!]
\centering
\begin{tikzpicture}[
  dot/.style={circle,fill=black,minimum size=3pt,inner sep=0pt,
            outer sep=-1pt},
  sv/.style={rectangle,
                                    dashed,
                                    minimum size=1cm},
  ]
  \node[dot] at (1,1) (A) {};
  \node[dot] at (1,2) (B) {};
  \node[dot] at (2,1) (C) {};
  \node[dot] at (2,2) (D) {};
  \node[dot] at (3,1) (E) {};
  \node[dot] at (3,2) (F) {};

  \node[dot, color=gray] at (4,1) (G) {};
  \node[dot, color=gray] at (4,2) (H) {};
  \node[dot, color=gray] at (0,1) (I) {};
  \node[dot, color=gray] at (0,2) (L) {};

  \node at (3.3,1.7) (M) {\footnotesize $p_2, p_2'$};
  \node at (-0.3,1) (M) {\footnotesize $p_1$};
  \draw (0.6,0.8) rectangle  (3.2,2.2);

  \draw[dotted] (1.8,0.6) rectangle  (4.2,2.4);

  \node at (3,0.4) (tag2) {$i+1$};
  \node at (2,0.4) (tag3) {$i$};

  \node at (-1,1) {$y$};
  \node at (-1,2) {$x$};

  \path[-, > = stealth, out=50,in=220, dashed] (I) edge node[above]{\small$\sim$} (F);
\end{tikzpicture}
\caption{Adjacent and overlapping symbolic valuations $\rho(i)$ (solid line) and $\rho(i+1)$ (dotted line) of length 3 ($-\lfloor\phi\rfloor = \lceil\phi\rceil = 1$), with points $p_2=(x,i,1)$ and $p_2' = (x,i+1,0)$ being locally equivalent. Both $p_1 \pfwd{\sim} p_2$ and $p_1 \pfwd{\sim} p_2'$ hold.}
\label{figure-example-concistency-fwd-II}
\end{figure}
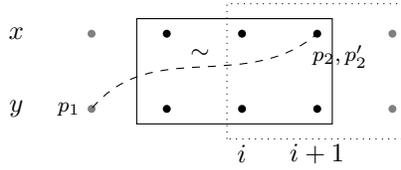

\begin{remark}
\label{rem:notless}
If $p_1 = \triple{x}{j}{h}$, $p_2 = \triple{y}{i}{m}$ and $p_1 \fwd p_2$, then 
it is $x(j+h) \leq y(i+m)$.
The other cases of $\pfwd{\sim}$ are similar. 
If $\sim$ is, respectively, $\prec, \approx, \succ, \succcurlyeq$, then relation between $x(j+h)$ and $y(i+m)$ is, respectively, $<, =, >, \geq$. 
If it is $p_1 \fwd p_2$, but not $p_1 \fwds p_2$, then along the path from $p_1$ to $p_2$ there are only arcs labeled with $\approx$, i.e. $p_1 \pfwd{\approx} p_2$, so $x(j+h) = y(i+m)$.
As a consequence, if it is $p_1 \fwd p_2$, but not $p_1 \fwds p_2$, then it is also $p_1 \bwd p_2$.
The dual properties hold for $\bwd$ and $\bwds$.
\end{remark}

Let $\rho= \alpha\beta^\omega \in SV(\phi)^\omega$ be an ultimately periodic symbolic model of $\phi$.
We need to introduce another notion of equivalence, which is useful for capturing properties of points of symbolic valuations in $\beta^\omega$, though it is defined in general.
More precisely, we consider two points $p,p' \in P_\rho$ as equivalent when they correspond to the same variable, in the same position of the symbolic valuation, but in symbolic valuations that are $i|\beta|$ positions apart, for some $i \geq 0$.
In fact, points in $\beta^\omega$ that are equivalent according to the definition below have the same properties concerning forward and backward paths.

\begin{definition}
\label{def:pequiv}
Two points $p,p'\in P_\rho$ are \emph{equivalent}, written $p
\equiv p'$, when $\var(p)=\var(p')$, $\sv(p')=\sv(p)+i|\beta|$ and
$\shift(p)=\shift(p')$, for some $i \in \Zed$.
\end{definition}

The main result of the section is Formula \eqref{our-nonex-C} on page \pageref{our-nonex-C}, which is based on a number of intermediate results that are presented in the following.
To test for the condition for the existence of arithmetic models of symbolic model $\rho = \alpha\beta^\omega$, 
one must represent infinite (possibly strict) forward and backward paths along $\rho$.
To this end, we devise a condition for the existence of infinite paths, 
resulting from iterating suffix $\beta$  infinitely many times.
Without loss of generality, in the following we consider ultimately periodic models $\rho = \alpha \beta^\omega$ in which $\alpha = \alpha' s$ and $\beta = \beta' s$, i.e., in which the last symbolic valuation of prefix $\alpha$ is the same as the last symbolic valuation of repeated suffix $\beta$.
We indicate by $k+1$ the length of $\alpha \beta$, and we number the symbolic valuations in $\alpha\beta$ starting from $0$, so that the last element in prefix $\alpha$ is in position $|\alpha| -1$, the first element in suffix $\beta$ is in position $|\alpha|$, and the last element of $\beta$ is in position $k$ (hence, $\rho(|\alpha|-1) = \rho(k) = s$, with $k = |\alpha\beta| -1$).
An infinite forward (resp. backward) path is represented as a cycle among
variables belonging to symbolic valuations $\rho(|\alpha|-1)$ and $\rho(k)$,
connected through relations $\fwd$ and $\fwds$ (resp. $\bwd$ and $\bwds$).
Intuitively, in $\rho$ there is an infinite (strict) forward path when there are two points $p,p'$ in $\alpha\beta$ -- with $p \neq p'$ -- such that $\sv(p) = |\alpha| -1$, $\sv(p') = k$, $p \equiv p'$, and $p \fwd p'$ ($p \fwds p'$).
Now, all results required to obtain Formula \eqref{our-nonex-C} equivalent to Property \ref{property-C} are provided.
%
%


We have the following property, which states that if in $\rho=\alpha\beta^\omega$ there is a finite forward path between two points $p,p''$ of the suffix $\beta^\omega$ with $p \equiv p''$, then there is also a finite forward path between $p$ and all points $p'$ between $p$ and $p''$ such that $p' \equiv p$.

\begin{lemma}
\label{lemma-equivpoints}
Let $\rho=\alpha\beta^\omega \in SV(\phi)^\omega$ be an ultimately periodic word, and $\beta = \beta' s' \beta''$ for some $\beta', \beta'' \in SV(\phi)^*, s' \in SV(\phi)$; let $i$ be the position of $s'$ in $\alpha\beta$ (so $\rho(i) = s'$).
Let $p_i,p_j$ any two points of $P_\rho$ such that $\sv(p_i) = i$, $\sv(p_j) = j$ and $p_i \equiv p_j$.
If $j > i + |\beta|$ and $p_i \pfwd{\sim} p_j$ (with $\sim \in \set{\preccurlyeq, \prec, \approx, \succcurlyeq, \succ}$), then it is also $p_i \pfwd{\sim} p'$, with $p \equiv p'$ and $\sv(p') = j - |\beta|$.
\end{lemma}

\begin{proof}
First of all, note that, since $p_i \equiv p_j$, it is $\rho(j-|\beta|) = \rho(j) = s'$.
\\
Let us consider the case $p_i \fwd p_j$.
Then, there exist three points $p_1,p_2,p_3$ such that:
\begin{enumerate}
\item either $p' \locfwd p_1$, or $p' \locbwd p_1$
\item $p_1 \fwd p_2$
\item $p_2 \locfwd p_j$
\item $p_i \fwd p_3$
\item $p_3 \fwd p_1$
\item either $p_3 \locfwd p'$, or $p_3 \locbwd p'$.
\end{enumerate}
\begin{figure}
\includegraphics[width=\columnwidth]{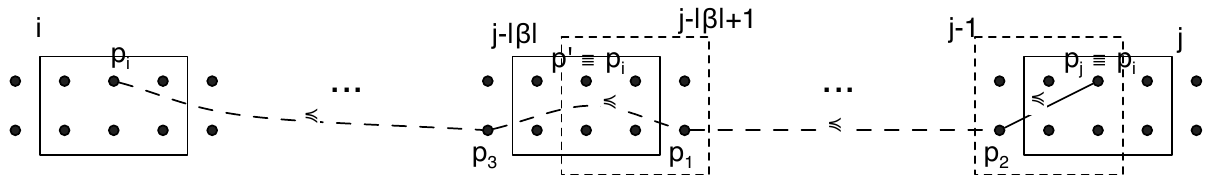}
\caption{Relations between symbolic valuations $i$ and $j$.}
\label{fig:LemmaLoopEquiv}
\end{figure}
Figure \ref{fig:LemmaLoopEquiv} exemplifies the conditions above.
We have two cases.
If $p' \locfwd p_1$, then, from conditions 2 and 3, and the definition of $\fwd$, we have $p' \fwd p_j$; since $p_i$, $p'$ and $p_j$ all belong to $\beta^\omega$ and are such that $p_i \equiv p' \equiv p_j$, then the same forward path between $p'$ and $p_j$ from which it descends $p' \fwd p_j$ can be iterated starting from $p_i$, because suffix $\beta^\omega$ is periodic.
Then, $p_i \fwd p'$.
If, instead, $p' \locbwd p_1$, then, by conditions 5 and 6, and definition of $\locbwd$ and $\fwd$, it is also $p_3 \fwd p'$; finally, by condition 4 and transitivity of $\fwd$, we have $p_i \fwd p'$.

The case $p_i \fwds p_j$ is similar, when one considers that, in addition to conditions 1-6, it must be $p_i \fwds p_3$, or $p_3 \fwds p_1$, or $p_1 \fwds p_2$, or $p_2 \locfwds p_j$.
If $p' \locfwd p_1$, then if it is $p_1 \fwds p_j$, it is also $p' \fwds p_j$, and the proof is as before.
If, instead, it is not $p_1 \fwds p_j$, then it must be $p' \fwds p_1$, otherwise from Remark \ref{rem:notless} it descends that the value of the variable in $p'$ is equal to the value in $p_j$, and in turn that the value of the variable in $p_i$ is equal to the value in $p_j$, thus contradicting $p_i \fwds p_j$.
If $p' \locbwd p_1$, then if $p_i \fwds p_1$ we have also $p_i \fwds p'$.
Otherwise, if it is not $p_i \fwds p_1$, then it must be $p_1 \fwds p_j$, and in this case it must also be $p' \locbwds p_1$ (hence also $p_i \fwds p'$), or the arc between $p_1$ and $p'$ is labeled with $=$, and we have that $p_i \fwd p'$, not $p_i \fwds p'$ (hence $p_i \pfwd{\approx} p'$ by Remark \ref{rem:notless}), and $p' \fwds p_j$, which yields a contradiction.

The proofs for cases $p_i \bwd p_j$, $p_i \bwds p_j$, and $p_i \pfwd{\approx} p_j$ are similar.
\end{proof}

We immediately have the following corollary, which states that a path looping through $p_i$ can be shortened to a single iteration.

\begin{corollary}
\label{cor:equivpoints}
Let $\rho=\alpha\beta^\omega \in SV(\phi)^\omega$, $p_i$ and $p_j$ as in Lemma \ref{lemma-equivpoints}.
Then it is also $p_i \pfwd{\sim} p'$, with $p \equiv p'$ and $\sv(p') = i + |\beta|$.
\end{corollary}

The following lemma shows that there is an infinite non-strict (resp. strict) forward path in $\rho = (\alpha' s)(\beta' s)^\omega$ if, and only if, there is an infinite non-strict (resp. strict) forward path that loops through symbolic valuation $s$.

\begin{lemma}
\label{lem:exofperfwdpath}
Let $\rho=\alpha\beta^\omega \in SV(\phi)^\omega$ be an ultimately periodic word, with $\alpha = \alpha' s$ and $\beta = \beta' s$.
In $\rho$ there is an infinite non-strict (resp. strict) forward path if, and only if, there is an infinite non-strict (resp. strict) forward path that contains a denumerable set of points $\{p_i\}_{i \in \Nat}$ of $P_\rho$ such that:
\begin{enumerate}
\item $\sv(p_0) = |\alpha| -1 = |\alpha'|$,
\item $p_i \equiv p_j$ and $\sv(p_i) < \sv(p_j)$ for all $i<j \in \Nat$,
\item $p_i \fwd p_{i+1}$ (resp. $p_i \fwds p_{i+1}$) for all $i \in \Nat$.
\end{enumerate}
\end{lemma}

\begin{proof}
Let us assume in $\rho$ there is an infinite non-strict forward path, and let $F = \{f_i\}_{i \in \Nat}$ be the points that it traverses (hence, it is $f_i \locfwd f_{i+1}$ for all $i$).
Note that $\sv(f_0)$ can be any, not necessarily $0$ or $|\alpha'|$.
Since suffix $\beta^\omega$ is periodic and each arc $\langle f_i, f_{i+1}\rangle$ in $F$ connects two points that, for Condition 1 of Definition \ref{def:lf_lv_peq}, dist at most $- \lfloor \phi \rfloor + \lceil \phi \rceil + 1$ from one another, then there must be a sequence of points $Q = \{q_i\}_{i \in \Nat}$ such that, for each $q_i \in Q$
\begin{itemize}
\item $\sv(q_{i+1}) > \sv(q_i) > |\alpha'|$
\item there is a point $f_j \in F$ such that $f_j$ is locally equivalent to $q_i$
\item $\rho(\sv(q_i)) = s$. 
\end{itemize}
In other words, $Q$ is made by points of $F$ (or locally equivalent ones) that belong to one of the instances of symbolic valuation $s$ in $\beta^\omega$.
For each $i \in \Nat$ it is $q_i \fwd q_{i+1}$.
Since the number of points in symbolic valuation $s$ is finite, there must be an element $q_{\bar{i}} \in Q$ such that an infinite number of points equivalent to $q_{\bar{i}}$ appear in $Q$.
In other words, there is a denumerable sequence $L = \{l_i\}_{i \in \Nat}$ such that
\begin{itemize}
\item $l_0 = q_{\bar{i}}$
\item for all $i$ it is $l_i \equiv q_{\bar{i}}$.
\end{itemize}
Again, for all $i$ it is $l_{i} \fwd l_{i+1}$ (also, it is $\sv(l_i) < \sv(l_{i+1})$).
Sequence $L$ is part of an infinite forward path that starts from $l_0$ and visits all $l_i$.
The desired sequence $\{p_i\}_{i \in \Nat}$ that satisfies conditions 1-3 is $L$ translated of $\sv(l_0) - |\alpha'|$, so that it starts from symbolic valuation in position $|\alpha'|$ (the translation is possible because of the periodicity of $\beta^\omega$).
Figure \ref{fig:LemmaTranslate} shows an example of translation.
\begin{figure}
\includegraphics[width=\columnwidth]{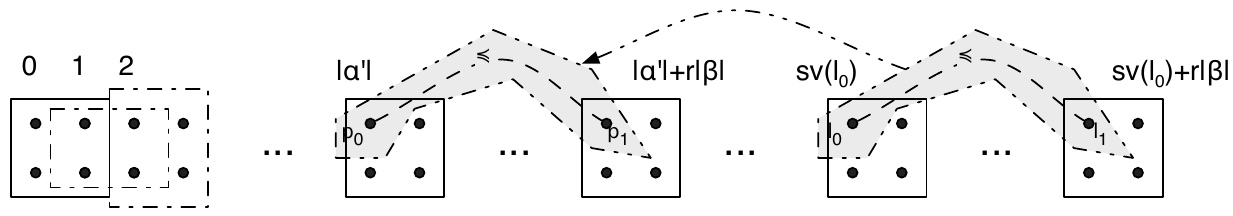}
\caption{Example of translation from $l_0$ to $p_0$.}
\label{fig:LemmaTranslate}
\end{figure}

The proof in case of strict infinite paths is similar.
\end{proof}

A similar lemma holds for backward paths.
We have the following result.

\begin{theorem}
\label{thm:LF}
Let $\rho=\alpha\beta^\omega \in SV(\phi)^\omega$ be an ultimately periodic word, with $\alpha = \alpha' s$ and $\beta = \beta' s$.
Then, there is a non-strict (resp. strict) infinite forward path in $\rho$ if, and only if,
there are two points $p,p'$ of $P_\rho$ such that $\sv(p) = |\alpha'|$, $\sv(p') = k$, $p \equiv p'$, and $p \fwd p'$ (resp. $p \fwds p'$).
\end{theorem}


\begin{proof}
We consider the case for non-strict forward paths, the case for strict ones being similar.

Assume in $\rho$ there is an infinite non-strict forward path; then, by Lemma \ref{lem:exofperfwdpath} there is also an infinite non-strict forward path that contains a denumerable set of points $\{p_i\}_{i \in \Nat}$ that satisfies conditions 1-3 of the lemma.
Then, from Corollary \ref{cor:equivpoints} we immediately have $p_0 \fwd p'$, with $p' \equiv p_0$ and $\sv(p') = |\alpha'|+|\beta| = k$ (recall that $|\alpha\beta| = k+1$).

Conversely, assume that there are two points $p,p'$ such that $p=\triple{x}{|\alpha'|}{h}$, $p'=\triple{x}{k}{h}$, $p \equiv p'$, and $p \fwd p'$.
By definition of $p\fwd p'$, there exists a finite number of points $p^1, p^2, \dots$ 
such that $p\locfwd p^1 \locfwd  p^2 \dots \locfwd p'$.
This forward path can be iterated infinitely many times, since $p\equiv p'$ and the suffix $\beta$ is repeated infinitely often.
Therefore, point $p$ and points equivalent to $p$ satisfy conditions 1-3 of Lemma \ref{lem:exofperfwdpath}.
By the same lemma, then, in $\rho$ there is an infinite non-strict forward path.
\end{proof}



Analogously,
we can prove the following version of
Theorem \ref{thm:LF} in case of backward paths.

\begin{theorem}
\label{thm:LB}
Let $\rho=\alpha\beta^\omega \in SV(\phi)^\omega$ be an ultimately periodic word, with $\alpha = \alpha' s$ and $\beta = \beta' s$.
Then, there is a non-strict (resp. strict) infinite backward path in $\rho$ if, and only if,
there are two points $p,p'$ such that $\sv(p) = |\alpha'|$, $\sv(p') = k$, $p \equiv p'$, and $p \bwd p'$ (resp. $p \bwds p'$).
\end{theorem}


Our condition for the non existence of an arithmetic model for symbolic model $\rho = \alpha' s (\beta' s)^\omega$ (with $|\alpha's\beta's| = k+1$) if formalized by Formula \eqref{our-nonex-C} below; it captures
Property \ref{property-C} and takes advantage of the previous Theorems \ref{thm:LF} and \ref{thm:LB}.

\begin{equation}
\label{our-nonex-C}
\begin{gathered}
\exists p_1 p_2 p_1' p_2'\left(
\begin{gathered}
p_1 \equiv p_2 \ \wedge \ p_1' \equiv p'_2 \ \wedge
\\
\sv(p_1)=\sv(p'_1)=|\alpha'| \ \wedge \ \sv(p_2)=\sv(p'_2)=k \ \wedge
\\
p_1 \fwd p_2 \ \wedge \ \ p'_1 \bwd p'_2    \ \wedge \ 
(p_1\fwds p_2 \ \vee \ \ p'_1\bwds p'_2)  \ \wedge \\
\ (p_1 \locfwds p_1' \vee {p_1'} \locbwds{p_1})
\end{gathered}
\right).
\end{gathered}
\end{equation}

In Formula \eqref{our-nonex-C} four conditions are defined, similar to those of Property \ref{property-C}.
Informally, Formula \eqref{our-nonex-C} says that:
\begin{enumerate}
\item there is an infinite forward path $f$ from $p_1$ (this derives from the fact that $p_1 \fwd p_2$, with $p_1 \equiv p_2$, $\sv(p_1)=|\alpha'|$, and $\sv(p_2) = k$);
\item there is an infinite backward path $b$ from $p_1'$ (from $p_1' \bwd p_2'$, with $p_1' \equiv p_2'$, where $\sv(p_1')=|\alpha'|$, and $\sv(p_2') = k$);
\item at least one between $f$ and $b$ is strict;
\item between $p_1$ and $p_1'$ there is an edge labeled with $<$.
\end{enumerate}
In particular, condition $4$ of Property \ref{property-C} is different from condition $4$ of Formula \eqref{our-nonex-C}.
In fact, the former one states that for each $i, j \in \N$, given a forward path $d$ and a backward path $e$,
whenever $d(i)$ and $e(j)$ belong to the same symbolic valuation (i.e., $|i-j| < -\lfloor\phi\rfloor + \lceil\phi\rceil +1$) there
is an edge labeled by $<$ from $d(i)$ to $e(j)$.
In other words, this means that point $p_d$ representing $d(i)$ and
point $p_e$ representing $e(j)$ are such that either $p_d \locfwds p_e$ or $p_e\locbwds p_d$.
The next theorem shows that the conditions are nevertheless equivalent when $\rho=\alpha\beta^\omega$.
In fact, whereas Property \ref{property-C} is defined for a general
$G_\rho$, Formula \eqref{our-nonex-C} is tailored to the finite representation
of ultimately periodic symbolic models $\rho=\alpha\beta^\omega$.

\begin{theorem}\label{thm:cond_equiv}
Over ultimately periodic symbolic models of the form $\alpha' s(\beta' s)^\omega$, with $\alpha,\beta\in SV(\phi)^*$ and $s \in SV(\phi)$, Property \ref{property-C} is equivalent to Formula \eqref{our-nonex-C}.
\end{theorem}
\begin{proof}
Let $\rho=\alpha' s(\beta' s)^\omega$ be an infinite symbolic model and assume that Formula \eqref{our-nonex-C} holds in $\alpha' s \beta' s$.
Therefore, by Theorems \ref{thm:LF} and \ref{thm:LB}, there exists a pair of points $p_1$ and $p_1'$, such that $\sv(p_1)=\sv(p_1')=|\alpha'|$, visited respectively by an infinite
forward path and an infinite backward path, where at least one of the two is strict (because $p_1\fwds p_2 \vee p'_1\bwds p'_2$ holds).
Since $p_1 \locfwds p_1' \vee p_1' \locbwds p_1$ holds, and $\equiv$ is a congruence for $\locfwds, \locbwds$, then also $p_2 \locfwds p_2'$ or $p_2' \locbwds p_2$ hold.
Now, consider any two points $u$ and $v$ in $\alpha' s \beta' s$, such that $\sv(u)=\sv(v)$ and $u$ (resp. $v$) belongs to the infinite strict forward (resp. backward) path from $p_1$ (resp. $p_1'$).
Then, it is $u \fwd p_2$, $v \bwd p_2'$, and $p_2 \locfwds p_2'$ or $p_2' \locbwds p_2$.
Hence, it is also $u \locfwds v$ or $v \locbwds u$, i.e., between $u$ and $v$ there is an edge labeled with $<$.

Conversely, assume Property \ref{property-C} holds; then, by Theorems \ref{thm:LF} and \ref{thm:LB} there are points $p_1,p_1',p_2,p_2'$ such that $\sv(p_1)=\sv(p_1')=|\alpha'|$, $\sv(p_2)=\sv(p_2')=k$, $p_1 \equiv p_2$, $p_1' \equiv p_2'$, $p_1 \fwd p_2$, $p_1' \bwd p_2'$, and $p_1 \fwd p_2 \vee p_1' \bwd p_2'$ hold.
From the proof of Theorem \ref{thm:LF}, point $p_1$ is equivalent to some point in the original forward path; similarly for point $p_1'$.
Then, since $p_1$ and $p_1'$ belong to the same symbolic valuation, by condition 4 of Property \ref{property-C}, they are connected through an edge labeled with $<$, i.e., $p_1 \locfwds p_1'$ or $p_1' \locbwds p_1$ hold.
\end{proof}

The next theorem extends Proposition
\ref{proposition-BSP-to-SAT-I} to constraint system IPC$^*$, which does
not benefit from the completion property.

\begin{proposition}\label{proposition-BSP-to-SAT-II}
Let $\phi\in$ CLTLB$(\D)$ and $\D$ be IPC$^*$.
Formula $\phi$ is $k$-satisfiable and Formula (\ref{our-nonex-C})
does \emph{not} hold if, and only if, thre exists a model $\sigma$ such that $\sigma,0 \models \phi$.
\end{proposition}
\begin{proof}
By Theorems \ref{theorem-eq-encod-run}, \ref{theorem-eq-enc-BSP}, and \ref{theorem-eq-BPS-run},
 $\phi$ is $k$-satisfiable if, and only if, formula $|\phi|_k$ is satisfiable;
in addition, when formula $|\phi|_k$ is satisfiable, it induces a model $\sigma_k$ and a sequence
$\alpha\beta$ of symbolic valuations of length $k$ representing an infinite sequence $\rho=\alpha\beta^\omega$ of symbolic valuations such that $\rho\symodels\phi$.
Since Formula (\ref{our-nonex-C}) does not hold, then by Theorem \ref{thm:cond_equiv} Property \ref{property-C} does not hold, hence, by Lemma \ref{lemma-condC}, $\rho$ admits a model $\sigma$ such that $\sigma,0\models\phi$.

Conversely, if formula $\phi$ is satisfiable, then automaton $\A_\phi$
recognizes models which satisfy condition $C$. Then,
a symbolic model $\alpha\beta^\omega\in \Lng(\A_\phi)$ and a model $\sigma_k,0\models_k\alpha\beta$ can be obtained as in the proof of
Proposition \ref{proposition-BSP-to-SAT-I}.
\end{proof}

\subsubsection*{Bounded Encoding of Formula \eqref{our-nonex-C}}
The encoding shown afterwards represents, by means of a finite representation, infinite -- strict and non strict -- paths over infinite symbolic models.
As before, we consider models $\rho = \alpha \beta^\omega$ where $\alpha = \alpha' s$ and $\beta = \beta' s$, and we consider the finite sequence of symbolic valuations $\alpha' s \beta' s$, of length $k+1$.
We indicate by $P_{\alpha\beta} \subset P_\rho$ the set of points of finite path $\alpha' s\beta' s$ (for all $p \in P_{\alpha\beta}$, it is $\sv(p) \in \interval{0}{k}$).
We use the points of $P_{\alpha\beta}$ to capture properties of $P_\rho$.
To encode the previous formulae into QF-EU$\D$ formulae, where
$\D$ is a suitable constraint system embedding $\Nat$ and having the
successor function plus order $<$, we rearrange the formulae above by
splitting information, which is now encapsulated in the notion of point, on variables and positions over the model.
Predicate $\Predlocfwds_{x,y}:\Nat^3\rightarrow \set{true,false}$ for all pairs $x,y \in V
\cup \const(\phi)$ (resp. $\Predlocfwd_{x,y}$) encodes relation $p_1\locfwds p_2$ (resp. $p_1\locfwd p_2$)
where $p_1=\triple{x}{j}{h}$ and $p_2=\triple{y}{j}{m}$.
\setlength{\extrarowheight}{1.5pt}
\[
\begin{array}{c|c}
0 \leq j \leq k \text{ and }h \leq m & 0 \leq j \leq k \text{ and }h > m\\
\hline
\Predlocfwds_{x,y}(j,h,m) \Leftrightarrow \sigma_k(j+h,x) < \sigma_k(j+m,y) & \neg\Predlocfwds_{x,y}(j,h,m)\\
\Predlocfwd_{x,y}(j,h,m) \Leftrightarrow \sigma_k(j+h,x) \leq \sigma_k(j+m,y) & \neg\Predlocfwd_{x,y}(j,h,m)
\end{array}
\]
for all $h,m \in \interval{\lfloor\phi\rfloor}{\lceil\phi\rceil}$.
The value of $\sigma_k(0+h,x)$ equals the value of
term $\alpha=\aY^{|h|} x$, for $h \in \interval{\lfloor\phi\rfloor}{-1}$, or
of term $\alpha=\aX^h x$, for $h \in \interval{0}{\lceil\phi\rceil}$.
For example, $\sigma_k(0+h,x)$ is $\bm{\alpha}(0)$, and $\sigma_k(k+h,x)$ is $\bm{\alpha}(k)$ (see
$|ArithConstraints|_k$ in Section \ref{section-encoding}).
Constants are implicitly included in the model.
For instance, if $5
\in \const(\phi)$ and $x \in V$  we have the following formulae
$\Predlocfwds_{x,5}(j,h,m) \iFF \sigma_k(j+h,x) < 5$
and $\Predlocfwds_{5,x}(j,h,m) \iFF 5 < \sigma_k(j+m,x) $.
When $x,y \in \const(\phi)$ then $\Predlocfwds_{x,y} \iFF x < y$ and $\Predlocfwd_{x,y}
\iFF x \leq y$ for all $0 \leq j \leq k$ and $h \leq m$; $\neg\Predlocfwds_{x,y}$ and $\neg\Predlocfwds_{x,y}$ for all $0 \leq j \leq k$ and $h > m$.

Relation $\fwds$ (resp. $\fwd$) is encoded by uninterpreted predicates $\Predfwds_{x,y}:\Nat^4
\rightarrow \set{true,false}$ (resp. $\Predfwd_{x,y}:\Nat^4
\rightarrow \set{true,false}$) for all pairs of variables $x,y \in
V\cup \const(\phi)$.
To build in practice $\fwds$ (resp., $\fwd$) through $\Predfwds$ (resp. $\Predfwd$), over points of the symbolic model $\alpha' s\beta' s$, we construct the transitive closure of $\Predfwds$ (resp. $\Predfwd$) explicitly.
Starting from $\rho(0)$, we propagate the information about relations $\locfwds$ and $\locfwd$  that are represented by $\Predlocfwds$ and $\Predlocfwd$ among
all points representing variables of model $\rho$.
In fact, it is immediate to show that $p_1 \fwds p_2$ holds if, and only if, there is a point $p$ such that either $p_1 \locfwds p$ and $p \fwd p_2$ or $p_1 \locfwd p$ and $p \fwds p_2$ (note that $p$ cannot be locally equivalent to both $p_1$ and $p_2$, but it can be locally equivalent to one of them).
Similarly for the other relations.
Figure \ref{figure-example-definition-F} provides a graphical representation for $\fwds$.
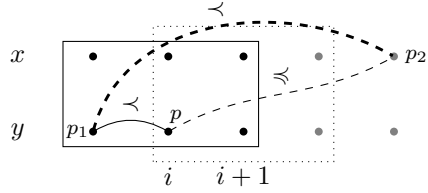
\begin{figure}[ht!]
\centering
\begin{tikzpicture}[
  dot/.style={circle,fill=black,minimum size=3pt,inner sep=0pt,
            outer sep=-1pt},
  sv/.style={rectangle,
                                    dashed,
                                    minimum size=1cm},
  ]
  \node[dot] at (1,1) (A) {};
  \node[dot] at (1,2) (B) {};
  \node[dot] at (2,1) (C) {};
  \node[dot] at (2,2) (D) {};
  \node[dot] at (3,1) (E) {};
  \node[dot] at (3,2) (F) {};
  \node[dot, color=gray] at (4,1) (G) {};
  \node[dot, color=gray] at (4,2) (H) {};
  \node[dot, color=gray] at (5,1) (I) {};
  \node[dot, color=gray] at (5,2) (L) {};
  \node at (5.3,2) (M) {\footnotesize $p_2$};
  \node at (0.8,1) (M) {\footnotesize $p_1$};
  \node at (2.1,1.2) (M) {\footnotesize $p$};
  \draw (0.6,0.8) rectangle  (3.2,2.2);

  \draw[dotted] (1.8,0.6) rectangle  (4.2,2.4);

  \node at (2,0.4) (tag3) {$i$};
  \node at (3,0.4) (tag4) {$i+1$};

  \node at (0,1) {$y$};
  \node at (0,2) {$x$};

  \path[-, > = stealth, out=75,in=150, dashed, line width=1.1pt] (A) edge node[above]{\small$\prec$} (L);
  \path[-, > = stealth, out=30,in=150] (A) edge node[above]{\small$\prec$}(C);
  \path[-, > = stealth, out=30,in=210, dashed] (C) edge node[above]{\small$\preccurlyeq$}(L);

\end{tikzpicture}
\caption{Adjacent symbolic valuations $\rho(i)$ (solid line) and $\rho(i+1)$ (dotted line) not covering both points $p_1 = (y,i,-1)$ and $p_2 = (x,j,h)$ (with $j>i$ and $-1 \leq h \leq 1$) of the model, with $p_1 \locfwds p$, $p \fwd p_2$ and $p_1 \fwds p_2$.}
\label{figure-example-definition-F}
\end{figure}
Formulae defining $\Predfwds_{x,y}$ and $\Predfwd_{x,y}$ are the following:
\begin{align}
\label{eq:FLTdef}
\Predfwds_{x,y}(j,h,i,m) & \Leftrightarrow
\left\lbrace
\begin{gathered}
\bigvee_{z \in V}
\bigvee_{u=\lfloor\phi\rfloor}^{\lceil\phi\rceil} \Predlocfwds_{x,z}(j,h,u) \wedge
\Predfwd_{z,y}(j,u,i,m) \vee \\
\bigvee_{z \in V}
\bigvee_{u=\lfloor\phi\rfloor}^{\lceil\phi\rceil} \Predlocfwd_{x,z}(j,h,u) \wedge \Predfwds_{z,y}(j,u,i,m)
\end{gathered}\right. \\
\label{eq:FLEQdef}
\Predfwd_{x,y}(j,h,i,m) & \Leftrightarrow
\bigvee_{z \in V}
\bigvee_{u=\lfloor\phi\rfloor}^{\lceil\phi\rceil}  \Predlocfwd_{x,z}(j,h,u) \wedge
\Predfwd_{z,y}(j,u,i,m)
\end{align}
for all $j,i
\in \interval{0}{k}$ with $j < i$ and for all $h,m \in
\interval{\lfloor\phi\rfloor}{\lceil\phi\rceil}$ such that 
$i+m - (j+h) > -\lfloor\phi\rfloor+\lceil\phi\rceil$, $h =
\lfloor\phi\rfloor$, $(x=z) \Rightarrow (h\neq u)$
and for all pairs $x,y \in V \cup \const(\phi)$.
When $j = i \in \interval{0}{k}$ and $h \leq m$, with $h,m \in
\interval{\lfloor\phi\rfloor}{\lceil\phi\rceil}$:
\[
\begin{aligned}
\Predfwds_{x,y}(j,h,j,m) &\Leftrightarrow \Predlocfwds_{x,y}(j,h,m) \\
\Predfwd_{x,y}(j,h,j,m) &\Leftrightarrow \Predlocfwd_{x,y}(j,h,m)
\end{aligned}
\]
When $j+h > i+m$:
\[
\begin{aligned}
\neg\Predfwds_{x,y}(j,h,i,m) \\
\neg\Predfwd_{x,y}(j,h,i,m)
\end{aligned}
\]
Figure \ref{figure-definition-F} shows how predicate $\Predfwds_{x,x}(i,0,j,1)$ is defined as
conjunction of local relation $\Predlocfwds_{x,y}(i,0,1)$ and of $\Predfwd_{y,x}(i,1,j,1)$.
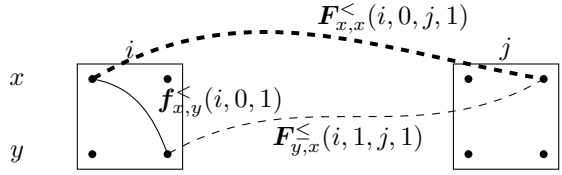
\begin{figure}[ht!]
\centering
\begin{tikzpicture}[
  dot/.style={circle,fill=black,minimum size=3pt,inner sep=0pt,
            outer sep=-1pt},
  sv/.style={rectangle,
                                    dashed,
                                    minimum size=1cm},
  ]
  \node at (1.5,2.4) {$i$};
  \node[dot] at (1,1) (A) {};
  \node[dot] at (1,2) (B) {};
  \node[dot] at (2,1) (C) {};
  \node[dot] at (2,2) (D) {};

  \draw (0.8,0.8) rectangle  (2.2,2.2);

  \node at (6.5,2.4) {$j$};
  \node[dot] at (6,1) (E) {};
  \node[dot] at (6,2) (F) {};
  \node[dot] at (7,1) (G) {};
  \node[dot] at (7,2) (H) {};

  \draw (5.8,0.8) rectangle  (7.2,2.2);

  \path[-, > = stealth', in=110,out=-10] (B) edge (C);
  \path[-, > = stealth', in=210,out=30, dashed] (C) edge (H);
  \path[-, > = stealth', in=170,out=30, dashed, line width=1.5pt] (B) edge (H);

  \node at (2.7,1.7) {$\Predlocfwds_{x,y}(i,0,1)$};
  \node at (0,1) {$y$};
  \node at (0,2) {$x$};

  \node at (4.4,1.2) {$\Predfwd_{y,x}(i,1,j,1)$};
  \node at (5,2.8) {$\Predfwds_{x,x}(i,0,j,1)$};
\end{tikzpicture}
\caption{Definition of $\Predfwds$ by local relations $\Predlocfwds$.}
\label{figure-definition-F}
\end{figure}

The following formula $|CongruenceConstraints|_k$ defines congruence classes of locally equivalent points for relations $\fwds, \fwd$ captured by predicates $\Predfwds$ and $\Predfwd$.
In fact, observe that, since from $p_1 \locfwd p_2$ we obtain $p'_1 \locfwd p'_2$, for all $p_1'$ (resp. $p_2'$) that is locally equivalent to $p_1$  (resp. $p_2'$), then, in general, the congruence extends to $\fwd$; i.e., from $p_1 \fwd p_2$ we obtain $p'_1 \fwd p'_2$ for all $p_1',p_2'$ locally equivalent to $p_1,p_2$.
An analogous argument holds for $\fwds$, $\bwd$ and $\bwds$.
\[
\begin{array}{c|c|c}
i \in\interval{1}{k}& m \in\interval{\lfloor\phi\rfloor}{\lceil\phi\rceil} & j \\
\hline
\Predfwds_{x,y}(j,h,i,m) \Leftrightarrow \Predfwds_{x,y}(j+1,h-1,i,m) & h\in
\interval{\lfloor\phi\rfloor+1}{\lceil\phi\rceil} &
\interval{0}{i-1}\\
\Predfwds_{x,y}(j,h,i,m) \Leftrightarrow \Predfwds_{x,y}(j-1,h+1,i,m) & h\in
\interval{\lfloor\phi\rfloor}{\lceil\phi\rceil-1} & \interval{1}{i}
\end{array}
 \]
\[
\begin{array}{c|c|c}
j\in\interval{0}{k-1} & h\in
\interval{\lfloor\phi\rfloor}{\lceil\phi\rceil} & i\\
\hline
\Predfwds_{x,y}(j,h,i,m) \Leftrightarrow \Predfwds_{x,y}(j,h,i+1,m-1) & m \in
\interval{\lfloor\phi\rfloor+1}{\lceil\phi\rceil} & i \in \interval{j}{k-1}\\
\Predfwds_{x,y}(j,h,i,m) \Leftrightarrow \Predfwds_{x,y}(j,h,i-1,m+1) & m\in
\interval{\lfloor\phi\rfloor}{\lceil\phi\rceil-1} & i \in \interval{j+1}{k}.
\end{array}
\]

Predicates $\Predlocbwds_{x,y}, \Predlocbwd_{x,y}$ for local backward paths $\locbwds, \locbwd$, predicates $\Predbwds_{x,y}, \Predbwd_{x,y}$ for backward paths $\bwds, \bwd$ and congruence among points are defined similarly.
For brevity, we only show the definition of $\Predlocbwds_{x,y}$ and $\Predlocbwd_{x,y}$, the others are straightforward.
\[
\begin{array}{c|c}
0 \leq j \leq k \text{ and }h \leq m & 0 \leq j \leq k \text{ and }h > m\\
\hline
\Predlocbwds_{x,y}(j,h,m) \Leftrightarrow \sigma_k(j+h,x) > \sigma_k(j+m,y) & \neg\Predlocbwds_{x,y}(j,h,m)\\
\Predlocbwd_{x,y}(j,h,m) \Leftrightarrow \sigma_k(j+h,x) \geq \sigma_k(j+m,y) & \neg\Predlocbwd_{x,y}(j,h,m)
\end{array}
\]
for all $h,m \in \interval{\lfloor\phi\rfloor}{\lceil\phi\rceil}$.
When both $x,y \in const(\phi)$ then $\Predlocbwds_{x,y}(j,h,m) \iFF x > y$ and $\Predlocbwd_{x,y}(j,h,m)
\iFF x \geq y$ for all $0 \leq j \leq k$ and $h \leq m$; $\neg \Predlocbwds_{x,y}(j,h,m)$ and $\neg \Predlocbwd_{x,y}(j,h,m)$ for all $0 \leq j \leq k$ and $h > m$.



Finally, the condition of existence defined by Formula \eqref{our-nonex-C} is encoded by the following QF-EU$\D$
formula.
The condition is parametric with respect to a pair of variables $x,x' \in V \cup \const(\phi)$.
The condition is meaningful only if $x \neq x'$ and if either $x \notin \const(\phi)$ or $x'\notin
\const(\phi)$.
In fact, a constant value never generates a strict (forward or
backward) path; therefore, two constants can not satisfy the condition
of non-existence of an arithmetical model.
Formula $C_{x,x'}$ below captures the existence in $\rho(|\alpha'|)$ of a strict relation $<$
between two points, one of a forward and one of backward path, which involve variables $x$ and $x'$.
Variable $\bm{loop}$ has already been introduced in Section \ref{qf-euf-encoding} and defines the position where, in $\alpha\beta$, suffix $\beta$ starts (as already explained $|\alpha'| = \bm{loop}-1$).
\[
C_{x,x'} := \bigvee_{h,h' \in \interval{\lfloor\phi\rfloor}{\lceil\phi\rceil} }
\left(
\begin{gathered}
\left(
  \begin{gathered}
  \Predfwd_{x,x}(\bm{loop}-1,h,k,h) \wedge
  \Predbwds_{x',x'}(\bm{loop}-1,h',k,h')
  \\ \vee \\
  \Predfwds_{x,x}(\bm{loop}-1,h,k,h) \wedge
  \Predbwd_{x',x'}(\bm{loop}-1,h',k,h')
  \end{gathered}
\right)
\\ \wedge \\
\Predlocfwds_{x,x'}(\bm{loop}-1,h,h') \vee \Predlocbwds_{x',x}(\bm{loop}-1,h',h)   \\
\end{gathered}
\right).
\]

In Formula $C_{x,x'}$, we use explicitly points that were symbolically represented in Formula \eqref{our-nonex-C}: $p_1 = (x,loop-1,h)$, $p'_1 = (x',loop-1,h')$, $p_2 = (x,k,h)$, $p'_2 =(x',k,h')$.
It is immediate to see that formula $\Predlocfwds_{x,x'}(\bm{loop}-1,h,h') \vee \Predlocbwds_{x',x}(\bm{loop}-1,h',h)$ encodes $p_1 \locfwds p_1' \vee p_1' \locbwds p_1$ of Formula \eqref{our-nonex-C} and formula $\Predfwd_{x,x}(\bm{loop}-1,h,k,h) \wedge \Predbwds_{x',x'}(\bm{loop}-1,h',k,h')$ , encodes $p_1 \fwd p_2 \wedge p_1' \fwd p_2' \wedge p_1 \fwds p_2$ (similarly for formula $\Predfwds_{x,x}(\bm{loop}-1,h,k,h) \wedge \Predbwd_{x',x'}(\bm{loop}-1,h',k,h')$).

%
%
The existence condition of an arithmetical model is captured by
the formula:
\begin{equation}\label{qf-euf-condition-C}
\bigwedge_{\footnotesize \begin{array}{c}x,x' \in V \cup \const(\phi)\\ x\neq x', x \notin \const(\phi) \lor x' \notin \const(\phi)\end{array}} \neg  C_{x,x'}
\end{equation}

Given a CLTLB$(\textrm{IPC}^*)$ formula $\phi$, 
the satisfiability of $\phi$ is reduced to the satisfiability of the following QF-EU($\D$) formula:
\begin{equation}\label{qf-euf-encoding}
|\phi|_k \land \eqref{qf-euf-condition-C}.
\end{equation}
If Formula \eqref{qf-euf-encoding} is unsatisfiable, then either $\phi$ does not admit symbolic models, or none of its symbolic models admits arithmetic models.
Conversely, if Formula \eqref{qf-euf-encoding} is satisfiable, then there is a symbolic model $\rho$ of $\phi$ for which condition \eqref{qf-euf-condition-C} holds, 
hence $\rho$ admits an arithmetic model and $\phi$ is satisfiable.

\subsection{Simplifying the condition of existence of arithmetical models}
\label{sec:simplifications}

In this section, we relax the condition of existence of an arithmetical
model $\sigma$ for sequences of symbolic valuations of CLTLB(IPC$^*$) formulae.
%
%
In fact, Property \ref{property-C} is stronger than necessary in those cases in which not all variables appearing in a formula $\phi$ are compared against each other.
Consider for example the following formula
\begin{equation}
\label{eq:xinc-ydec}
\G(x < \aX x \wedge \neg(y < \aX y))
\end{equation}
which enforces strict increasing monotonicity for variable $x$ and decreasing monotonicity for variable $y$.
Figure \ref{figure-xinc-ydec} shows a symbolic model for Formula \eqref{eq:xinc-ydec} which does not
admit arithmetic model, as it does not satisfy Property \ref{property-C} (in fact, the strict forward path that visits all points $\{\triple{x}{i}{0}\}_{i \in \Nat}$ and the strict backward path that visits all points $\{\triple{y}{i}{0}\}_{i \in \Nat}$ are such that, for all $i$, $\triple{x}{i}{0} \locfwds \triple{y}{i}{0}$).
\begin{figure}[!htb]
\centering
\includegraphics[scale=0.7]{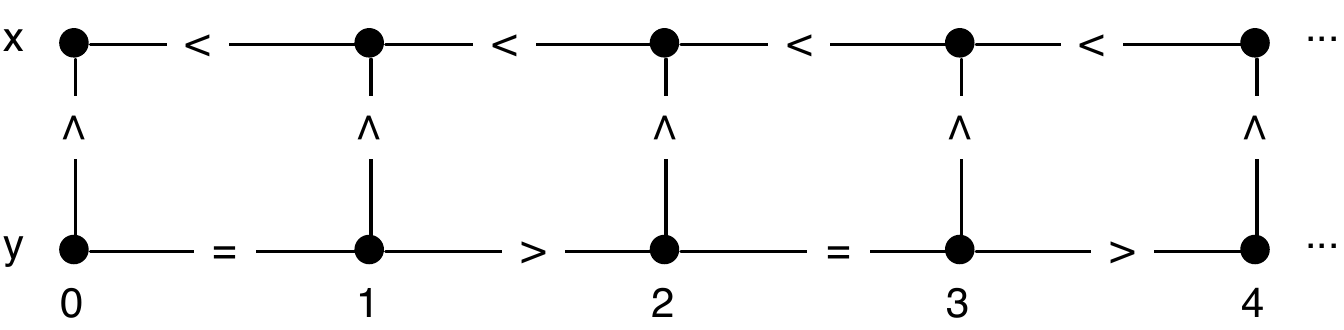}
\caption{A symbolic model for Formula \eqref{eq:xinc-ydec} that does not admit an arithmetical model.}
\label{figure-xinc-ydec}
\end{figure}
%
However, in Formula \eqref{eq:xinc-ydec} $x$ and $y$ are not compared, neither directly, nor indirectly, so
if we disregard the relations between them in the symbolic model of Figure \ref{figure-xinc-ydec},
and produce an assignment of the variables that only respects the relations between variables that are actually compared in
the formula (i.e., $x$ with itself, and $y$ with itself) we obtain an arithmetic model for Formula \eqref{eq:xinc-ydec}.
Figure \ref{figure-xinc-ydec-weak} shows a ``weaker'' version of the symbolic model of Figure \ref{figure-xinc-ydec}, one that
is more concise to encode into QF-EU($\D$) formulae than the maximally consistent one, as it does not contain any comparison between unrelated terms.
\begin{figure}[h!]
\centering
\includegraphics[scale=0.7]{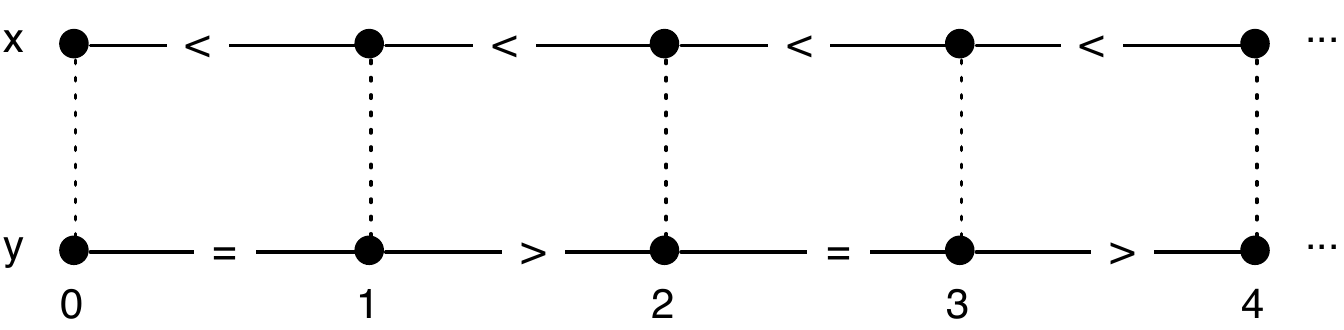}
\caption{A weak symbolic model for Formula \eqref{eq:xinc-ydec}.}
\label{figure-xinc-ydec-weak}
\end{figure}

To characterize sequences of symbolic valuations which do not
take into account relations among variables that are not compared with each other in
a formula $\phi$, we first remark that $\phi$ induces a finite partition
$\set{V_1,\dots, V_h}$ of set $V$ such that $x,y \in V_i$ if and only if
there is an IPC$^*$ constraint $R(\aX^i x, \aX^j y)$ occurring in $\phi$,
for some $i,j \in \Zed$ (where we write $\aX^{-n}$, with $n > 0$, instead of $\aY^n$).
Then, we introduce the notions of \emph{weak symbolic valuation} and
of \emph{sequence of weak symbolic valuations}.

\begin{definition}
\label{def:weaksv}
Given a symbolic valuation $sv \in SV(\phi)$, its \emph{weak}
version $\weak{sv}$ is obtained by removing from $sv$ all relations $R(\aX^i x,
\aX^j y)$ where $x \in V_l$ and $y \in V_t$ with $l\neq t$.
We similarly define the weak version $\weak{\rho}$ of a sequence $\rho$ of symbolic valuations.
\end{definition}

Given a CLTLB(IPC$^*$) formula $\phi$, we indicate with $SV_{w}(\phi)$ the set of all its weak symbolic valuations.
A weak symbolic model $\weak{\rho} \in SV_w(\phi)^\omega$ of $\phi$ is a
sequence of weak symbolic valuations such that $\weak{\rho},0\symodels \phi$.
Given $\rho \in SV(\phi)^\omega$ and its weak version $\weak{\rho}$, $G_{\weak{\rho}}$ is the subgraph of $G_\rho$ ontained by removing all arcs between points $p = \triple{x}{j}{h}$, $p' = \triple{y}{i}{m}$ such that $x \in V_l$, $y \in V_t$, and $l \neq t$.

The next lemma shows that focusing on weak symbolic valuations is enough to determine whether symbolic models for $\phi$ exist or not.

\begin{lemma}\label{lemma-weak-symbolic}
Let $\phi$ be a CLTLB(IPC$^*$) formula.
Given $\rho \in SV(\phi)^\omega$ such that $\rho, 0 \symodels \phi$, it is also $\weak{\rho}, 0 \symodels \phi$.
Conversely, given  a sequence $\nu \in SV_w(\phi)$ of weak symbolic valuations, if $\nu, 0 \symodels \phi$, then for any $\rho \in SV(\phi)$ such that $\weak{\rho} = \nu$ it is also $\rho, 0 \symodels \phi$.
\end{lemma}

\begin{proof}
Assume that $\rho \symodels \phi$.
We only need to focus on the base case, as the inductive one is trivial.
For all $i$ and $R(\alpha_1,\alpha_2)$ occurring in $\phi$, $\rho,i\symodels R(\alpha_1,\alpha_2)$ if, and only if, $R(\alpha_1,\alpha_2) \in \rho(i)$.
Since $R(\alpha_1,\alpha_2)$ occurs in $\phi$ then, by Definition \ref{def:weaksv}, it is also $R(\alpha_1,\alpha_2) \in \weak{\rho}(i)$, hence $\weak{\rho},i\symodels R(\alpha_1,\alpha_2)$.

The converse case is similar.
If $\nu \in SV_w(\phi)$ is such that $\nu, 0 \symodels \phi$, then for all $i$ and $R(\alpha_1,\alpha_2)$ that occurs in $\phi$ it is $\nu,i\symodels R(\alpha_1,\alpha_2)$ if, and only if, $R(\alpha_1,\alpha_2) \in \nu(i)$; in addition, for any $\rho$ such that $\weak{\rho} = \nu$ we have $R(\alpha_1,\alpha_2) \in \rho(i)$ if, and only if, $R(\alpha_1,\alpha_2) \in \nu(i)$.
Finally, $\nu,i\symodels R(\alpha_1,\alpha_2)$ implies $\rho,i\symodels R(\alpha_1,\alpha_2)$.
\end{proof}

We have the following variant of Lemma \ref{lemma-condC}, which defines a condition of existence of arithmetical models for symbolic ones that is checked on their weak countersparts.

\begin{lemma}
\label{prop:weak-C}
Let $\phi$ be a CLTLB(IPC$^*$) formula.
Given an ultimately periodic, locally consistent sequence $\rho \in SV(\phi)^\omega$ of symbolic valuations, if there is $\sigma : \Zed \times V \to D$ such that $\sigma, 0 \models \rho$, then Property \ref{property-C} holds for graph $G_{\weak{\rho}}$.
Conversely, if $\nu \in SV_w(\phi)^\omega$ is an ultimately periodic, locally consistent sequence of weak symbolic valuations such that Property \ref{property-C} holds for graph $G_{\nu}$, then there are $\sigma$, $\rho$ such that $\weak{\rho} = \nu$ and $\sigma, 0 \models \rho$.
\end{lemma}

\begin{proof}
If there is $\sigma$ such that $\sigma, 0 \models \rho$ then, by Lemma \ref{lemma-condC}, Property \ref{property-C} holds for $G_\rho$.
Since $G_{\weak{\rho}}$ is a subgraph of $G_\rho$, a fortiori Property \ref{property-C} holds for $G_{\weak{\rho}}$.

Conversely, if Property \ref{property-C} holds for $G_\nu$, then each set of variables $V_i$, with $i \in \set{1..h}$, in which $V$ is partitioned induces an ultimately periodic sequence $\nu_{V_i}$ of symbolic valuations that only include constraints on $V_i$, such that its graph $G_{\nu_{V_i}}$ is not connected to any other graph $G_{\nu_{V_j}}$, for $j \neq i$.
Then, Lemma \ref{lemma-condC} can be applied to $\nu_{V_i}$, which then admits an arithmetic model $\sigma_{V_i} : \Zed \times V_i \to D$.
By definition, each $\sigma_{V_i}$ assigns a different set of variables, so the complete arithmetic model $\sigma$ is simply the union of all $\sigma_{V_i}$.
By Lemma \ref{lemma-A}, $\sigma$ induces a sequence of symbolic valuations $\rho$, and $\sigma, 0 \models \rho$, $\weak{\rho} = \nu$ by construction.
\end{proof}

Thanks to Lemmata \ref{lemma-weak-symbolic} and \ref{prop:weak-C}, in Formula \eqref{our-nonex-C} and in the corresponding QF-EU($\D$) encoding of Formula \eqref{qf-euf-condition-C} we can focus only on relations between points that belong to the same set $V_i$.

\section{Complexity and Completeness}
\label{sec:ComplAndCompl}

\subsection*{Complexity}\label{subsection-complexity}

In the following we provide an estimation of the size of the formulae
constituting the encoding of Section \ref{section-encoding}, including,
where they are needed, the constraints of Section \ref{subsection-A_C}.

The encoding of Section \ref{section-encoding} is linear in the size of the
formula $\phi$ (and of the bound $k$).
In fact, if $m$ is the total number of subformulae and $n$ is the total number of
temporal operators $\U$ and $\R$ occurring in $\phi$, 
the QF-EU$\D$ encoding
requires $n+1$ integer variables (one each for $\bm{loop}$ and
the $\bm{j_\psi}$'s) and $m$ unary 
predicates (one for each subformula in $cl(\phi)$).

The total size of the formulae in Section
\ref{subsection-A_C} is polynomial in bound $k$, in the cardinality of the set of
variables and constants, and in the size of symbolic valuations.
In fact, the encoding of the condition for the existence of an
arithmetical model requires a QF-EU$\triple{\Nat}{<}{=}$ formula of
size quadratic in the length $k$, cubic in the number $|V|$ of variables,
and double quadratic in the size of symbolic valuations.

Let $\lambda$ be the size $\lambda=\lceil\phi\rceil - \lfloor\phi\rfloor+1$
of symbolic valuations and $V'$ be the set $V \cup \const(\phi)$.
The total number of non-trivial predicates $\Predlocfwd_{x,y},\Predlocfwds_{x,y}$
(resp. $\Predlocbwd_{x,y},\Predlocbwds_{x,y}$), i.e., those where $h\leq
m$, is defined by the following parametric formula (where $a,b$ are the sets to which $x,y$ belong, respectively):
\begin{align*}
N(a,b) & = (k+1)\sum_{i=1}^\lambda |a|\cdot \left( (\lambda-i) + (|b|-1)\cdot(\lambda-i+1) \right) \\
       & = (k+1)\left(|a||b| \frac{\lambda(\lambda+1)}{2} - |a|\lambda\right).
\end{align*}
Each predicate has fixed dimension and the number of non-trivial ones results from the sum of the following three cases:
\begin{itemize}
\item $x,y \in V$, which is $N(V,V)$
\item $x \in V$, $y \in \const(\phi)$, which is $N(V,\const(\phi))$
\item $x \in \const(\phi)$, $y \in V$, which is $N(\const(\phi),V)$.
\end{itemize}
that is bounded by $N_{local}=N(V',V')\leq (k+1)|V'|^2 \lambda^2$.

To compute the size of formulae defining $\Predfwd_{x,y},\Predfwds_{x,y}$
(resp. $\Predbwd_{x,y},\Predbwds_{x,y}$) we first determine the number of pairs of points
for which $\Predfwd_{x,y}(j,h,i,m)$ is not trivially false.
The following function $N_{p,p'}$
\begin{align*}
N_{p,p'} & = |V'|\sum_{i=\lfloor\phi\rfloor}^{k+\lceil\phi\rceil}|V'|(k+\lceil\phi\rceil-i)
= |V'|^2\sum_{i=0}^{k+\lambda-1}i = |V'|^2\frac{(k+\lambda-1)(k+\lambda)}{2}  \\
         & \leq |V'|^2(k+\lambda)^2
\end{align*}
corresponds to the number of pairs of points $p,p'$ that generate non-trivial predicates $\Predfwd_{x,y}$, $\Predfwds_{x,y}$
(resp. $\Predbwd_{x,y}$, $\Predbwds_{x,y}$) because their position is such that
$\sv(p_1)+\shift(p_1) \leq \sv(p_2)+\shift(p_2)$ (resp.
$\sv(p_1)+\shift(p_1) \geq \sv(p_2)+\shift(p_2)$).
We compute the size of (non-trivial) formulae \eqref{eq:FLTdef}-\eqref{eq:FLEQdef} defining $\Predfwds_{x,y},\Predfwd_{x,y}$ (and
$\Predbwds_{x,y},\Predbwd_{x,y}$) by counting the number of subformulae involved
in their definition.
We consider only the case for $\Predfwds_{x,y}$ because the others have the same
(worst) complexity.
Each Formula \eqref{eq:FLTdef} involves, in the worst case (i.e., for points that do not belong to the
same symbolic valuation), $|V|-1$ variables $z\in V$ with respect to
$\lambda$ different positions $u$.
Then, an instance of \eqref{eq:FLTdef} requires at most $(|V|-1)\lambda$ disjuncts.
The upper bound for the total size of all formulae defining
predicates $\Predfwd_{x,y},\Predfwds_{x,y}$ (resp. $\Predbwd_{x,y},\Predbwds_{x,y}$) is
$$
N_{far}=N_{p,p'}2(|V|-1)\lambda \leq \lambda|V||V'|^2(k+\lambda)^2\leq \lambda|V'|^3(k+\lambda)^2.
$$

The analysis of formulae $|CongruenceConstraints|_k$ shows that each
point belongs to $\lambda$ symbolic valuations (e.g., if $\lceil\phi\rceil = 0$, $\lfloor\phi\rfloor = -1$,
then $\lambda = 2$, and points $(x,4,1)$ and $(x,5,0)$ correspond to the same element),
and for all pairs
$p_1,p_2$ we define the consistency of the definition of predicate $\Predfwds_{x,y}$ among the $\lambda$ points corresponding to
$p_1$ and the $\lambda$ points corresponding to $p_2$.
Therefore, we need at most
$$
N_{CC}=4|V'|^2\sum_{i=1}^{k}\lambda^2 i \leq 4\lambda^2|V'|^2k^2
$$
constraints $|CongruenceConstraints|_k$, where each constraint has
fixed dimension.

Finally, predicate $C_{x,x'}$ appears in Formula \eqref{qf-euf-condition-C}
once for each of the $|V'||V|\lambda^2$ pairs of points ${x,x'}$.
In addition, each instance of $C_{x,x'}$ has $\lambda^2$ disjuncts, one for each possible pair $h,h' \in \interval{\lfloor\phi\rfloor}{\lceil\phi\rceil}$.
Therefore, the total size of Formula \eqref{qf-euf-condition-C} is $N_C=|V||V'|\lambda^4$.

Finally, the complete set of formulae that we require to capture the existence
condition of arithmetical models over discrete domains has the following total size:
\[
\begin{gathered}
4N_{local} + 4N_{far} + 4N_{CC}+N_C \leq \\
4(k+1)|V'|^2 \lambda^2 + 4\lambda|V'|^3(k+\lambda)^2+16\lambda^2|V'|^2k^2 + |V||V'|\lambda^4.
\end{gathered}
\]
In conclusion, for a given formula $\phi$, the parameters $\lambda$ and $|V'|$ are fixed, hence the size is $\mathcal{O}(k^2)$.

\subsection*{Completeness}\label{section-completeness}

Completeness has been studied in depth for Bounded Model Checking.
Given a state-transition system $M$, a temporal logic property $\phi$ and a bound $k>0$, 
BMC looks for a
witness of length $k$ for $\neg \phi$.
If no witness exists then length $k$ may be increased and BMC may be reapplied.
In principle, the process terminates when a witness is found or when $k$ reaches a
value, the \emph{completeness threshold} (see Definition \ref{completeness}), which guarantees that if
no counterexample has been found so far, then no counterexample disproving
property $\phi$ exists in the model.
For LTL it is shown that a completeness threshold always exists;
\cite{CKOS04} shows a procedure to estimate an over-approximation of the value,
 by satisfying a formula representing
the existence of an accepting run of the product automaton $M\times B_{\neg
  \phi}$, where $B_{\neg  \phi}$ is the B\"uchi automaton for $\neg\phi$ and $M$ is the system to be verified.

In \cite{BFRS11} we have already given a positive answer to the
problem of whether
there exists a completeness threshold for the satisfiability problem for
CLTLB($\D$), provided that ultimately periodic symbolic models of the form
$\alpha\beta^\omega $ of CLTLB($\D$) formulae admit an arithmetic model.
By the results of Section \ref{subsec:completion} this occurs when the constraint system $\D$
has the completion property, or when condition $C$ holds.
%
%
In \cite{BFRS11} we used a mixed automata- and logic-based approach to show
how completeness can be achieved for the satisfiability problem.
In that approach automata
$\A_C$ and $\A_\ell$ described in Section \ref{section-symbolic} are represented through CLTLB($\D$) formulae $\phi_{\A_C}$ and
$\phi_{\A_\ell}$, respectively, described below.
More precisely, formula $\phi_{\A_C}$ captures the runs of  automaton $\A_C$, and similarly for $\phi_{\A_\ell}$ and $\A_\ell$.
Then, checking the satisfiability for $\phi$ is reduced to studying a finite amount of  $k$-satisfiability problems of formula $\phi \wedge \phi_{\A_C} \wedge \phi_{\A_\ell} $ for increasing values of $k$.
Automaton $\A_\ell$ recognizes sequences of locally consistent symbolic valuations, so its runs are the models of formula
$
\phi_{\A_\ell} :=  \G(\bigvee_1^m sv_i).
$
Since the bounded representation of formulae (see Section \ref{section-encoding}) is not
contradictory (i.e., two consecutive symbolic valuations are satisfiable
when they are locally consistent),
 the previous formula exactly represents words of $\Lng(\A_\ell)$.
Formula $\phi_{\A_C}$, instead, is derived from automaton $\A_C$, by means of
the translation in \cite{SC85}.
Automaton $\A_C$ is built by complementing
automaton $\A_{\neg C}$ \cite{S88}, recognizing the
complement language of $\Lng(\A_C)$, which is
obtained according to the procedure proposed in \cite{DD07}. 
%
%
Finally, to check the satisfiability of $\phi$ we verify whether formula $\phi \; \wedge
\; \phi_{\A_C} \wedge \ \phi_{\A_\ell}$ 
is $k$-satisfiable, with $k \in \N$.
%
The existence of a finite completeness threshold for the procedure above is a consequence of
the existence of automaton $\A_\phi$ (see Section \ref{section-symbolic}) recognizing symbolic models of
$\phi$, and of Lemma \ref{lemma-condC} and Proposition \ref{prop-completion}.
Let $rd(\A_\phi)$ be the recurrence diameter of $\A_\phi$, i.e., the longest loop-free path in the automaton that starts from an initial state \cite{KS03}.
Then, if formula $\phi \; \wedge\; \phi_{\A_C} \wedge \ \phi_{\A_\ell}$ is not $k$-satisfiable for all $k \in [1, rd(\A_\phi)+1]$, then there
is no ultimately periodic symbolic model $\rho$ such that both $\rho,0
\symodels \phi$ and there exists an arithmetic model $\sigma$ with $\sigma,0 \models \rho$.
Hence, formula $\phi$ is unsatisfiable.
Otherwise, we have found an ultimately periodic symbolic model $\rho$
of length $k>0$ which admits an arithmetic  model $\sigma$.
From the $k$-bounded solution, we have a symbolic model $\rho =
\alpha\beta^\omega$ and its bounded arithmetic model $\sigma_k$.
The infinite model $\sigma$ is built from $\sigma_k$ by iterating infinitely
many times the sequence of symbolic valuations in $\beta$.
Therefore, the completeness bound for BSP of CLTLB($\D$)
formulae is defined by the recurrence diameter of $\A_{\phi}$.

Thanks to the results of the previous sections, we can simplify the method
presented in \cite{BFRS11}.
We avoid the construction of automaton $\A_{\neg C}$ through Safra's method and the construction of
set $SV(\phi)$.
In particular, we take advantage of the definition of $k$-bounded
models of $\phi$.
By Lemma \ref{lemma-B},
a finite sequence $\sigma_k$  of $\D$-valuations
induces a unique locally consistent
sequence of symbolic valuations $\rho$, such that $\sigma_k,i\models
\rho(i)$, for all $i \in \interval{0}{k}$.
Therefore, we do not need to precompute set $SV(\phi)$ of symbolic valuations and formula $\phi_{\A_\ell}$ is no longer needed to obtain
a finite locally consistent  sequence of symbolic valuations.
If $\phi$ is a formula of  CLTLB$(\D)$ and $\D$ has
the completion property, we can simply solve $k$-satisfiability problems for
$\phi$ instead of $\phi \wedge \phi_{\A_\ell}$;
when $\D$ does not have the completion property,
Formula \eqref{our-nonex-C} allows us to avoid the construction of
$\A_C$.
In the first case, by Theorems \ref{theorem-eq-encod-run} -- \ref{theorem-eq-BPS-run} and Proposition \ref{proposition-BSP-to-SAT-I} $|\phi|_k$ is
satisfiable if, and only if, there is an ultimately periodic run $\alpha\beta^\omega$ which
is recognized by automaton $\A_s\times\A_\ell$.
In the second case, Proposition
\ref{proposition-BSP-to-SAT-II} guarantees that $|\phi|_k$ is
satisfiable and Formula \eqref{our-nonex-C} does not hold if, and only if, $\phi$ is satisfiable.
Therefore, model $\alpha\beta^\omega$ obtained by solving the
$k$-satisfiability problem belongs to the language
recognized by automaton $\A_s\times\A_\ell$ and also to the one recognized by $\A_C$.

The completeness property still holds without the explicit representation of automata
$\A_\ell$ and $\A_C$ in the formula we check for satisfiability.
Since the role of Formula \eqref{our-nonex-C} is to filter, by eliminating edges in the automaton, some of the symbolic models of $\phi$ which, in turn, by Theorems \ref{theorem-eq-encod-run} -- \ref{theorem-eq-BPS-run} correspond to the runs of automaton $\A_s\times\A_\ell$, the completeness threshold for our decision procedure can be over-approximated by the recurrence diameter of $\A_s\times\A_\ell$, which is at most exponential in the size of $\phi$.
Since the number of control states of automaton $\A_s$ is at most $\mathcal{O}(2^{|\phi|})$,
a rough estimation for the completeness threshold is given by the value
$|SV(\phi)|\cdot 2^{|\phi|}$.
The number of symbolic valuations $|SV(\phi)|$ is, in the worst case, exponential in the
size of formula $\phi$ \cite{DD07}.

\section{Applications of k-bounded satisfiability}\label{sec-applications}

The decision procedure described in this paper has been implemented in our bounded
satisfiability checker \zot{} ({\small \url{http://zot.googlecode.com}}).
The $ae^2Zot$ plug-in of \zot{} 
solves $k$-satisfiability for CLTLB over
Quantifier-Free Presburger arithmetic ($\QFP$), of 
which IPC$^*$ is a fragment, but it also supports the constraint system $\triple{\Real}{<}{=}$.
Even if constraint systems like IPC$^*$, or fragments thereof, do not provide a
counting mechanism (provided, for instance, through the addition
of functions like $+$ in $\QFP$), they can still be used to represent an abstraction of a richer
transition system.
In fact, functions like addition, or in general
relations over counters which embed a counting mechanism, make the
satisfiability problem of CLTLB undecidable (see 
\cite[Section 9.3]{DD07}).

We next examplify the use of the CLTLB
logic to specify and verify systems behavior, thus highlighting the applicability of the approach.

We use CLTLB over $\triple{D}{<}{=}$ to specify a sorting process of a sequence of fixed length $N$ 
of values in $D$.
Let $\mathbf{v} \in D^N$ be the (initial) vector that we want to sort
and $\mathbf{a} \in D^N$ be the vector during each step of sorting.
We write $\mathbf{v}(i)$ for the $i$-th component of $\mathbf{v}$, $1 \leq i \leq N$.
Notice that we will use the notation $\mathbf{a}(i)$,
which, strictly speaking, is not a CLTLB term;
however, since the length of the array is fixed, we can use $N$
variables $a_i$ to represent the elements of $\mathbf{a}$, one for each $\mathbf{a}(i)$.
Then, in the following, if $\mathbf{a}(i)$ is replaced with $a_i$, one obtains CLTLB$\triple{D}{<}{=}$ formulae.
We define a set of formulae representing a sorting process which swaps
unsorted pairs of values at some nondeterministically chosen position in the vector
(we report here only the most relevant formulae).
A variable $p \in \interval{0}{N-1}$ stores the position of elements which are
a candidate pair for swapping; i.e., $p=i$ means that 
$\mathbf{a}(i)$ is swapped with 
$\mathbf{a}(i+1)$, while $p=0$ means that no elements are swapped (0 is not a position of the vector). 
A nondeterministic algorithm can swap arbitrarily two elements in
$\interval{1}{N}$; then, the only constraint on variable $p$ is that it is $0 \leq p < N$, i.e.:
$\G(p < N \wedge p \geq 0)$.
An unsorted pair of values is indexed by a nonzero value of $p$:
$$
\G\left(\bigwedge_{i \in \interval{1}{N-1}} p=i \Rightarrow \mathbf{a}(i)>\mathbf{a}(i+1)\right).
$$
A swap between two adjacent positions of $\mathbf{a}$ is formalized
by the following formula:
$$
\G\left(\bigwedge_{i \in \interval{1}{N-1}} p = i \Rightarrow
  \aX\mathbf{a}(i) = \mathbf{a}(i+1)  \wedge \aX\mathbf{a}(i+1) = \mathbf{a}(i)\right).
$$
Vector $\mathbf{a}$ is unchanged when no pairs are candidate for
swapping: $\G(p = 0 \Rightarrow \bigwedge_{i \in
  \interval{1}{N}}(\mathbf{a}(i) = \aX\mathbf{a}(i)))$.
Various properties of the algorithm have been verified through the $ae^2Zot$ plugin of the \zot{} tool, e.g., whether there exists a way to sort array $\mathbf{a}$ within $k$ steps (with $k$ the verification bound), which is formalized by the following formula:
$$
\F\left(\bigwedge_{i \in \interval{1}{N-1}} (\mathbf{a}(i) \leq \mathbf{a}(i+1)) \wedge
  \bigwedge_{ i \in \interval{1}{N}}\bigvee_{j \in \interval{1}{N}} (\mathbf{a}(i) = \mathbf{v}(j))\right).
$$

\section{Related works}
\label{section-related-works}

For some constraint system $\D$ more expressive
than IPC$^*$, the future fragment
CLTL($\D$) can encode runs of Minsky machines, a class of Turing-equivalent
two-counter automata.
Minsky machines are finite state automata endowed with two nonnegative integer
counters $c_1,c_2$ which can be
either incremented or decremented by $1$ and tested against $0$ over
transitions.
For example, to represent increment and decrement instructions 
the grammar of formulae $\xi$ of IPC$^*$ can be enriched with
formulae of the form $x < y + d$, where $d
\in D$ and $x,y$ are variables (these correspond to difference logic -- $\DL$ -- constraints).
Hereafter, we write $\cltl{a}{b}(\D)$ to denote the language of CLTL
formulae such that the cardinality of $V$ is $a$ and 
$\lceil\phi\rceil$ is $b$ (while $\lfloor\phi\rfloor$ is of course 0).

The first undecidability result for the satisfiability of CLTL
is given by Comon and Cortier \cite[Theorem 3]{CC00} who show that halting runs of a Minsky
machine can be encoded into $\cltl{3}{1}(\DL)$ formulae
where one auxiliary counter encodes control states of the system
labeling instructions.
Therefore, the satisfiability problem for $\cltl{3}{1}(\DL)$ is
$\Sigma_1^1$-hard.
The authors suggest a way to regain decidability by means of a syntactic
restriction on formulae including the $\U$ temporal operator.
The ``flat'' fragment of $\cltl{\omega}{1}(\DL)$ consists of CLTL
formulae such that subformula $\phi$ of $\phi\U\psi$ is $\top$,
$\perp$ or a conjunction $\zeta_1\wedge\dots \wedge \zeta_m$ where $\zeta_i
\in \DL$.
The fragment has a nice
correspondence with a special class of counter system (flat relational
counter system) with B\"uchi
acceptance condition, for which
the emptiness problem is decidable.
%
Satisfiability is undecidable also in the case of $\cltl{1}{2}(\DL)$
and $\cltl{2}{1}(\DL)$.
In fact, even though $\cltl{1}{2}(\DL)$ has only one variable,
it is expressive enough to encode runs of
Minsky machines:
models of $\cltl{1}{2}(\DL)$ formulae can represent
counter $c_1$ at even positions and counter $c_2$ at odd positions.
The recurrence problem for
nondeterministic Minsky machines, which is 
$\Sigma_1^1$-hard \cite{Alur&Henzinger94},
can be reduced to the satisfiability problem for
$\cltl{1}{2}(\DL)$, which then results $\Sigma_1^1$-hard.
From the previous undecidability results, the satisfiability problem for the CLTL
language over two integer variables
$\cltl{2}{1}(\DL)$ is $\Sigma_1^1$-hard.
In fact, formulae of $\cltl{1}{2}(\DL)$ can be syntactically translated to
formulae of $\cltl{2}{1}(\DL)$ by means of a map
$f$ such
that $\phi$ belonging to $\cltl{1}{2}(\DL)$ is satisfiable if, and
only if, $f(\phi)$ belonging to $\cltl{2}{1}(\DL)$ is satisfiable.
%
Both the languages $\cltl{1}{2}(\DL)$ and $\cltl{2}{1}(\DL)$ are also
$\Sigma_1^1$-complete by reducing
the $\Sigma_1^1$-hard model-checking problem to satisfiability.


The satisfiability (and model-checking) problem for CLTL over
structure $\triple{D}{<}{=}$ with $D\in\{\N,
\Zed, \Q, \Real\}$ is studied in \cite{DD07}, and for
$\IPC^*$ in \cite{DG07}.
Decidability of the satisfiability problem for the above cases is
shown by means of an automata-based approach similar to the standard case
for LTL.
Satisfiability for $\cltl{\omega}{\omega}(\IPC^*)$ and
$\cltl{\omega}{\omega}(<,=)$ over $\Nat, \Zed, \Q, \Real$ is obtained by
Demri and Gascon in \cite{DG05} by reducing it to the emptiness
of B\"uchi automata.
Given a CLTL formula $\phi$, it is possible to define an automaton
$\A_\phi$ such that $\phi$ is satisfiable if, and only if, $\Lng(\A_\phi)$ is not
empty.
Since the emptiness of $\Lng(\A_\phi)$ in the considered structures is
decidable with $\pspace$ upper bound (in the
dimension of $\phi$), then the satisfiability problem is also decidable with
the same complexity.
We remark that the notion of symbolic valuation in that work is different from the one we adopted in Definition~\ref{def-sv}.
Since the procedure is purely symbolic, constraints representing equality relation $x=d$ and constraints of the form $x\equiv_c d$, with $d,c\in D$, are explicitly considered, as no arithmetical model $\sigma$ is available.
A symbolic valuation is defined there as a triple $\langle S_1, S_2, S_3\rangle$ where $S_1$
is a maximally consistent set of $\D$-constraints over $\terms(\phi)$ and $\const(\phi)$; $S_2$ is
a set of constraints of the form $x=d$, and $S_3$ is a set of constraints $x\equiv_K c$, where constant $K$ is the least common multiple of constants occurring in constraints $x \equiv_c y$ and $x \equiv_c y + d$.

Sch\"{u}le and Schneider \cite{Schuele&Schneider07} provide a general
algorithm to decide bounded 
$\LTL(L)$ model-checking problems of infinite state systems where
$L$ is a general underlying logic.
An $\LTL(L)$ formula $\phi$ is translated into
an equivalent B\"uchi automaton $\A_\phi$ which is symbolically 
represented by means of a structure defining its transition relation and
acceptance condition.
Then, the $\LTL(L)$ model-checking problem is reduced to the $\mu$-calculus
model-checking problem modulo $L$, i.e., a verification of a fixpoint
problem for a given Kripke structure with respect to symbolic
representations of $\A_\phi$ and the underlying language $L$.
Whenever properties are neither proved nor disproved over finite
computations, their truth value  
can not be defined.
For this reason, the authors adopt a three-valued logic to evaluate formulae
whose components may have undefined
value.
Bounded model-checking is performed essentially by computing approximate
fixpoint sets of the desired formula and by checking whether the initial condition 
is a subset of such set of states.
The work of \cite{Schuele&Schneider07} is based on previous results presented in
\cite{Schule&Schneider04}, which defines
a hierarchy of B\"uchi automata (and, therefore, temporal
formulae) for which infinite state bounded model-checking is complete.
The specification language of \cite{Schule&Schneider04} is the
quantifier-free fragment of Presburger 
LTL, $\LTL(\PA)$, with past-time temporal modalities.
The bounded model-checking problem is defined with respect to Kripke
structures $\triple{S}{I}{R}$ and it
is solved by means of a reduction to the satisfiability of Presburger formulae.
In general, acceptance conditions of B\"uchi automata, requiring that
some states are visited infinitely often,
can not be handled immediately by bounded
approaches which do not 
consider ultimately periodic models used, for instance, in the bounded model-checking
approach of Biere et al. \cite{BCCZ99} or in the encoding of B\"uchi
automata of de Moura et al. \cite{dMRS02}.
Therefore, Sch\"{u}le and Schneider follow a different approach, tailored
to bounded verification, and
focus on the analysis of some classes of $\LTL$ 
formulae, denoted TL$_\F$ and TL$_\G$,  such that the corresponding
B\"uchi automaton has a simpler
accepting condition which does not involve
infinite computations.
TL$_\F$ and TL$_\G$ are the sets of $\LTL$ 
formulae such that each occurrence of a
weak/strong temporal operator is negative/positive and
positive/negative, respectively.
$\LTL$ formulae are then represented symbolically by an automaton which
is built using the method proposed by Clarke
et al. in \cite{CGH94} rather than using the Vardi-Wolper 
construction \cite{vw}.

Reducing the model-checking problem to Presburger satisfiability is a rather
standard approach when dealing with infinite-state systems.
Demri et al. in  \cite{DFGvD10} show how to solve the $\LTL(\PA)$ model-checking
problem for the class of 
{\em admissible} counter systems, which are finite state automata endowed
with variables over $\Zed$ whose transitions are labeled by Presburger
formulae.
In \cite{DFGvD10} the authors study the decidability of the model-checking
problem for admissible counter systems
with respect to the first-order CTL$^*$ language over Presburger formulae.

Hodkinson et al. study decidable fragments of first-order temporal logic in
\cite{HWZ00}.
Although some axiomatizations of first-order temporal logic are known,
various incompleteness results induce the authors to study useful
fragments with expressiveness between that of propositional and of first-order temporal logic.
Hodkinson et al. are interested in studying the satisfiability problem
and they do not consider the model-checking problem, which requires a formalism
defining the interpretation of first-order variables over time.
In other words, variables do not vary over time and their temporal
behavior is not relevant.
The languages investigated by the authors are obtained by
restricting both the first-order part and the temporal part.

Bultan et al. present a symbolic model checker for
analyzing programs with unbounded integer domains \cite{BGP99}.
Programs are defined by an event-action language where atomic events
are expressed by Presburger formulae over programs variables $V$.
Semantics of programs is defined in terms of infinite transition
systems
where the states are determined by the values of variables.
The specification language is a CTL-like temporal logic enriched with
Presburger-definable constraints over $V$.
Solving the CTL model-checking problem involves the computation of
least fixpoints over sets of programs states:
the abstract interpretation of Cousot and Cousot \cite{CC77} 
provides a method to compute approximation of
fixpoints.
Model-checking is done conservatively: the 
approximation technique admits false negatives, i.e., the solver may 
indicate that a property does not hold when it actually does.
Programs are analyzed symbolically by means of symbolic execution
techniques
and they are represented by means of Presburger-definable transition systems
where Presburger formulae represent
symbolically the transition relation and the set of program states.
Then, the state space is partitioned to reduce the complexity of
verification and to obtain decidability for some classes of temporal properties, such as
reachability ones.
Experimental results, based on the standard Bakery algorithm
and the Ticket mutual-exclusion algorithm, show the effectiveness of the
method when 
verification involves a mutual exclusion requirement. 


\section{Conclusions and further developments} \label{section-conclusions}

In this paper, we provide a procedure for deciding the satisfiability problem for CLTLB over some suitable constraint systems.
The main advantage of our approach is that it allowed us to implement the first effective tool 
based on SMT-solvers for those logics.
On one side, this method illustrates a new way to solve verification problems of formalisms dealing with variables ranging over infinite domains and having  an inherent notion of discrete time as that of LTL.
Instead of building an automaton for proving the satisfiability of a formula (which would be unfeasible in practice), we devise a direct method to construct one of its accepting runs which define a model for the formula.
On the other hand, our framework constitutes a foundation for defining extensions to handle different temporal formalisms.
In \cite{BRS13} we use the same approach presented in this paper to allow for the use of variables whose behaviour is restricted to clocks \cite{AD94} into CLTLB$\triple{\Real}{<}{=}$.
A clock is a nonnegative variable accumulating the time elapsed since the position where it was reset to 0 and that can be used to measure time between two discrete positions.
When dealing with clocks, it is common to consider a uniform progression of time; the time elapsing is unique for all the clocks that are updated by the same value at each position of the discrete model.
In \cite{BRS13} we prove the decidability and the complexity of the satisfiability problem for the CLTLB logic endowed with a finite set of clocks, and we provide an effective implementation to solve it through SMT-solvers which extends the one presented in this work.

In \cite{BRS13-GandALF} we devise a reduction from MITL formulae interpreted over continuous time to CLTLB formulae with clocks.
Since the reduction guarantees the equisatisfiability between the MITL formula and the resulting translation into CLTLB formulae, the satisfiability problem for the former logic can actually be solved. 



\bibliographystyle{model2-names}
\bibliography{bibliografia}

\end{document}